\documentclass[letterpaper,11pt]{article}

\usepackage[margin = 1.1in]{geometry}

\usepackage{mathptmx}       
\usepackage{helvet}         
\usepackage{courier}        
\usepackage{type1cm}        
\usepackage{makeidx}         
\usepackage{graphicx}        
\usepackage{multicol}        
\usepackage[bottom]{footmisc}
\usepackage{algorithmic,algorithm}
\usepackage{amssymb,amsmath,amsthm}

\usepackage{array}

\makeatletter
\newcommand{\thickhline}{%
    \noalign {\ifnum 0=`}\fi \hrule height 1pt
    \futurelet \reserved@a \@xhline
}
\newcolumntype{"}{@{\hskip\tabcolsep\vrule width 1pt\hskip\tabcolsep}}
\makeatother
\newtheorem{theorem}{Theorem}[section]

\newtheorem{corollary}[theorem]{Corollary}
\newtheorem{lemma}[theorem]{Lemma}

\theoremstyle{definition}
\newtheorem{definition}[theorem]{Definition}

\theoremstyle{remark}
\newtheorem{remark}[theorem]{Remark}

\newcommand{\opt}{\mathrm{OPT}}

\usepackage{url}

\usepackage[font=footnotesize]{caption}
\usepackage[font=footnotesize]{subcaption}

\urldef{\smith}\url{stephen.smith@uwaterloo.ca}
\urldef{\fata}\url{efata@uwaterloo.ca}
\urldef{\alamdari}\url{s26hosse@uwaterloo.ca}

\usepackage{epsfig,color}

\graphicspath{{./fig/}}

\makeindex             

\begin{document}

\title{Min-Max Weighted Latency Walks: Approximation Algorithms for Persistent
Monitoring in Vertex-Weighted Graphs}

\title{Persistent Monitoring in Discrete Environments:  Minimizing \\the Maximum Weighted Latency Between Observations\thanks{A preliminary version of this paper appeared in the Proceedings of the 2012 Workshop on the Algorithmic Foundations of Robotics~\cite{SA-EF-SLS:12}.}}
\author{Soroush Alamdari\thanks{Cheriton School of Computer Science, University of Waterloo, Waterloo ON N2L 3G1, Canada. email: \alamdari.}  \qquad Elaheh Fata$^{\ddagger}$ \qquad Stephen L. Smith\thanks{Department of Electrical and Computer Engineering, University of Waterloo, Waterloo ON N2L 3G1, Canada. email: \fata; \smith}
}
\date{}
%
%
\maketitle

\begin{abstract}
In this paper, we consider the problem of planning a path for a robot to monitor a known set of features of interest in an environment. We represent the environment as a graph with vertex weights and edge lengths.  The vertices represent regions of interest, edge lengths give travel times between regions, and the vertex weights give the importance of each region.  As the robot repeatedly performs a closed walk on the graph, we define the weighted latency of a vertex to be the maximum time between visits to that vertex, weighted by the importance (vertex weight) of that vertex.
Our goal is to find a closed walk that minimizes the maximum weighted latency of any vertex. We show that there does not exist a polynomial time algorithm for the problem.  We then provide two approximation algorithms; an $O(\log n)$-approximation algorithm and an $O(\log \rho_G)$-approximation algorithm, where $\rho_G$ is the ratio between the maximum and minimum vertex weights.  We provide simulation results which demonstrate that our algorithms can be applied to problems consisting of thousands of vertices, and a case study for patrolling a city for crime.
\end{abstract}

\section{Introduction}

An emerging application area for robotics is in performing long-term monitoring tasks.  Some example monitoring tasks include 1) environmental monitoring such as ocean sampling~\cite{RNS-MS-SLS-DR-GSS:10f-JFR}, where autonomous underwater vehicles sense the ocean to detect the onset of algae blooms; 2) surveillance~\cite{NM-ES-KM:11}, where robots repeatedly visit vantage points in order to detect events or threats, and; 3) infrastructure inspection such as power-line or manhole cover inspection~\cite{TT-CR-PJ:11-arxiv}, where spatially distributed infrastructure must be repeatedly inspected for the presence of failures.  For such tasks, a key problem is the high-level path planning problem of determining robot paths that visit different parts of the environment so as to efficiently perform the monitoring task.  Since some parts of the environment may be more important than others (e.g., in ocean sampling, some regions are more likely to experience an algae bloom than others), the planned path should visit regions with a frequency proportional to their importance.

In this paper, we cast such long-term monitoring tasks as an optimization problem on a graph with vertex weights and edge lengths: the \emph{min-max latency walk problem}.  The vertices represent regions (or features) of interest that must be repeatedly observed by a robot.  The edge lengths give travel times between regions, and the vertex weights give the importance of each region.  A vertex is observed by the robot once it is reached.  Given a robot \emph{walk}\footnote{We refer to a robot path as a walk to be consistent with the terminology in graph theory~\cite{JAB-USRM:08}, where a path is a sequence of unique vertices, while a walk may repeat vertices.} on the graph, the \emph{weighted latency} of a vertex is the maximum time between visits to that vertex, weighted by the importance (vertex weight) of that vertex.  We then seek to find a walk that minimizes the maximum weighted latency over all vertices. In an infrastructure task, this would be akin to minimizing the expected number of failures that occur in any region prior to a robot visit.


\emph{Prior work:}  The problem of continuously (or persistently) monitoring an environment using mobile robots has been studied in the literature under several names, including continuous sweep coverage; patrolling; persistent surveillance; and persistent monitoring.  The basis of all of these problems is sweep coverage~\cite{HC:01}, where a robot must move through the environment so as to cover the entire region with its sensor.  Variants of sweep coverage include on-line coverage~\cite{YG-ER:03}, where the robot has no \emph{a priori} information about the environment, and dynamic coverage~\cite{IIH-DMS:07}, where each point in the environment requires a pre-specified ``amount'' of coverage.

One approach to persistent monitoring has been to focus on randomized approaches in discrete environments.  These works frequently use the name of continuous sweep coverage. In~\cite{AT-MJ-DEJ-RMM:05}, a continuous coverage problem is considered where a sensor must continually survey regions of interest by moving according to a Markov chain.  In~\cite{GC-AS:11} a similar approach to continuous coverage is taken and a Markov chain is used to achieve a desired visit-frequency distribution over a set of features.  In~\cite{EA-EK-NCM:12-arxiv}, the authors look at robots modelled by controlled Markov chains, and seek to persistently monitor regions while avoiding forbidden regions.

The other main approach to persistent monitoring is to cast the problem as one in combinatorial optimization. This research typically falls under the name of patrolling or persistent surveillance.  In the most basic problem, a robot seeks to minimize the time between visits to each point in space.  For this problem, both centralized approaches~\cite{YC:04,YE-NA-GAK:07,NN-IK:08}, and distributed algorithms for multiple robots~\cite{fp-af-fb:09v} have been proposed. Recently, there has been work on minimizing the weighted time between visits to each region of interest, where the weight of a region captures its priority relative to other regions~\cite{SLS-MS-DR:10a, fp-jwd-fb:11h}.  These works focus on controlling robots along predefined paths.  In~\cite{SLS-MS-DR:10a}, the authors consider velocity controllers for persistent monitoring along fixed tours, while in~\cite{fp-jwd-fb:11h} the authors focus on coordination issues for multiple robots on fixed tours. An optimal control formulation for persistent monitoring in one-dimensional spaces is given in~\cite{CGC-XCD-XL:11}.

In our prior work~\cite{SLS-DR:09a}, we considered a specific case of the min-max latency walk problem on Euclidean graphs, where the graphs are constructed by distributing vertices in a Euclidean space according to a known probability distribution.  Under these assumptions, constant factor approximation algorithms are developed for the limiting case when the number of vertices becomes very large.

The combinatorial approaches taken in persistent monitoring problems often draw from solutions to vehicle routing~\cite{NC-JEB:84, GL:09} and dynamic vehicle routing (DVR) problems~\cite{FB-EF-MP-KS-SLS:09Proc}.  
The authors of~\cite{ES-NM:11} make the connection between multi-robot persistent surveillance and the vehicle routing problem with time windows~\cite{GL:09}.  In~\cite{TT-CR-PJ:11-arxiv}, the authors consider a preventative maintenance problem in which the input is a vertex and edge-weighted graph, as in the min-max latency walk problem, but the output is a path which visits each vertex exactly once.  More important vertices (i.e., those that are more likely to fail) should be visited earlier in the path.  The authors find a path by solving a mixed-integer program.  The min-max latency walk problem can be thought of as a generalization of preventative maintenance, where the maintenance and inspection should continually be performed.

Other work in persistent monitoring and surveillance includes~\cite{BB-JPH-JV:08, BB-JR-JPH-MAV-JV:10}, where a persistent task is defined as one whose completion takes much longer than the life of a robot. The authors focus on issues of battery management and recharging.  Such issues have also been recently considered in the context of persistent monitoring in~\cite{NM-SLS-SW:13}.

\emph{Contributions:}  There are four main contributions of this paper.  First, we introduce the general min-max latency walk problem and show that it is well-posed and that it is APX-hard.  Second, we provide results on the existence of optimal algorithms and approximation algorithms for the problem.  We show that in general, the optimal walk can be very long---its size can be exponential in the size of the input graph, and thus there cannot exist a polynomial time algorithm for the problem. We then show that there always exists a constant factor approximation solution that consists of a walk of size in $O(n^2)$, where $n$ is the number of vertices in the input graph. Third, we provide two approximation algorithms for the problem. Defining $\rho_G$ to be the ratio between the maximum and minimum vertex weights in the input graph $G$, we give an $O(\log \rho_G)$ approximation algorithm.  Thus, when $\rho_G$ is independent of $n$, we have a constant factor approximation. We also provide an $O(\log n)$ approximation which is independent of the value of $\rho_G$.  The algorithms rely on relaxing the vertex weights to be powers of $2$, and then planning paths through ``batches'' of vertices with the same relaxed weights. Fourth, and finally, we show in simulation that we can solve large problems consisting of thousands of vertices, and we demonstrate our algorithm on a case study of patrolling the city of San Francisco, CA for crime.

A preliminary version of this paper appeared as~\cite{SA-EF-SLS:12}.  Compared to the conference version, this version presents detailed proofs of all statements, new results on the existence of optimal finite walks, additional remarks and illustrative examples, and a new case study on patrolling in the simulations section.


\emph{Organization:} This paper is organized as follows.  In Section~\ref{sec:background}, we give the essential background on graphs and formalize the min-max latency walk problem.  In Section~\ref{sec:relaxation}, we present a relaxation of graph weights which allows for the design of approximation algorithms.  In Section~\ref{sec:existence_results}, we present results on the existence of constant factor approximations and some negative results on the required size of the walk.  In Section~\ref{sec:approx_algs}, we present two approximation algorithms for the problem.  In Section~\ref{sec:simulations}, we present large scale simulation data for standard TSP test-cases and perform a case study in patrolling for crime. Conclusions and the future directions are presented in Section~\ref{sec:conclusions}.

\section{Background and Problem Statement}
\label{sec:background}

In this section, we review graph terminology and define the min-max latency walk problem.

\subsection{Background on Graphs}
\label{sec:graph_background}

\paragraph{Vertex-weighted and edge-weighted graphs:}  The vertex set and edge set of a graph $G$ are denoted by $V(G)$ and $E(G)$, respectively, such that $E(G)\subseteq V(G)\times V(G)$. The number of vertices in a graph $G$, i.e., $|V(G)|$, is called the \emph{size of graph} $G$ and is denoted by $n$.
An edge in $E(G)$ is referred to as $(v_i,v_j)$ or $v_iv_j$. We consider only \emph{undirected} graphs, meaning $(v_i,v_j)\in E(G)$ if and only if $(v_j,v_i)\in E(G)$. An edge-weighted graph $G$ associates a weight $l(e) > 0$ to each edge $e\in E(G)$. We will refer to the weight of an edge as its \emph{length}. A vertex-weighted graph $G$ associates a \emph{weight} $\phi(v) \in [0,1]$ to each vertex $v\in V(G)$. Given a graph $G$ and a set $V'\subseteq V(G)$, the graph $G[V']$ is the graph obtained from $G$ by removing the vertices of $G$ that are not in $V'$ and all edges incident to a vertex in $V(G)\setminus V'$. Throughout this paper, all referenced graphs are both vertex-weighted and edge-weighted and therefore we omit the explicit reference. Also, without loss of generality, we assume that there is at least one vertex in $V(G)$ with weight $1$, as in our applications weights can be scaled so that this is true. We define $\rho_G$ to be the ratio between the maximum and minimum vertex weights, i.e.,
\[
\rho_G:=\max_{v_i,v_j \in V(G)}\frac{\phi(v_i)}{\phi(v_j)}.
\]

\paragraph{Walks in graphs:} A walk of \emph{size} $k$ in a graph $G$ is a sequence of vertices $(v_1,v_2,\ldots,v_{k})$ such that for any $1\leq i\le k$ we have $v_i\in V(G)$ (with the possibility that $v_i=v_j$ for some $1\le i,j\le k$) and there exists an edge $v_jv_{j+1} \in E(G)$ for each $1\le j<k$.
A walk is \emph{closed} if its beginning and end vertices are the same. {The \emph{length} of a walk $W=(v_1,\ldots,v_k)$, denoted by $l(W)$, is defined as the sum of the length of edges of graph $G$ that appear in that walk, i.e.,
\[
l(W)=\sum_{i=1}^{k-1}{l(v_iv_{i+1})}.
\]
Given a walk $W=(v_1,\ldots,v_k)$, and integers $i\leq j \leq k$, the \emph{sub-walk} $W(i,j)$ is defined as the subsequence of $W$ given by $W(i,j)=(v_i,v_{i+1},\ldots,v_j)$.  Given the walks $W_1,W_2,\ldots,W_l$, the walk $W=[W_1,W_2,\ldots,W_l]$ is the result of \emph{concatenation} of $W_1$ through $W_l$, while preserving order.

\paragraph{The traveling salesman problem:} A \emph{tour} of a graph $G$ is a closed walk $T=(v_1,v_2,\ldots,v_{n+1})$ that visits all $n$ vertices in $G$ such that $v_1=v_{n+1}$.  Thus, a tour visits each vertex exactly once and then returns to its start vertex.  The \emph{Traveling Salesman Problem} (TSP) is to find a tour in $G$ of minimum length.  We refer to a solution of the TSP as a \emph{TSP tour} and we denote it by $\mathrm{TSP}(G)$.  We denote the first $n$ vertices on a TSP tour of $G$ as $\text{TSP-Path}(G)$.  Thus, $\text{TSP-Path}(G)$ is a walk that visits each vertex exactly once.  The length of $\text{TSP-Path}(G)$ is upper bounded by the length of $\text{TSP}(G)$.

\paragraph{Infinite walks in graphs:}  An \emph{infinite walk} is a sequence of vertices, $(v_1,v_2,\ldots)$, such that there exists an edge $v_iv_{i+1} \in E(G)$ for each $i\in \mathbb{N}$. We say that a walk $W$ \emph{expands} to an infinite walk $\Delta(W)$ if $\Delta(W)$ is constructed by an infinite number of copies of $W$ concatenated together, i.e., $\Delta(W)=[W,W,\ldots]$. It can be seen that for any walk $W$, there exists a unique expansion to an infinite walk. The \emph{kernel} of an infinite walk $W$, denoted by $\delta(W)$, is the shortest walk such that $W$ is the \emph{expansion} of $\delta(W)$. It is easy to observe that there are infinite walks for which a finite size kernel does not exist. For such an infinite walk $W$, we define $\delta(W)$ to be $W$ itself.

\subsection{The Min-Max Latency Walk Problem}

Let $G=(V,E)$ be a graph with edge lengths $l$ and vertex weights $\phi$.  In a robotic monitoring application, the graph $G$ can be obtained as a discrete abstraction of the environment, with vertices corresponding to regions of the environment and edge lengths corresponding to travel distances (or times) between regions.  The vertex weights on the graph give the relative importance of the regions for the monitoring task.   Given an infinite walk $W$ for the robot in $G$, we define the \emph{latency} of vertex $v$ on walk $W$, denoted by $L(W,v)$, as the maximum length of the sub-walk between any two consecutive visits to $v$ on $W$. The latency of a vertex $v$ corresponds to the maximum time between observations of the region represented by $v$.

Then, we can define the weighted latency, or \emph{cost of a vertex} $v\in V(G)$ on the walk $W$ to be
\[
C(W,v) := \phi(v)L(W,v).
\]
The \emph{cost of an infinite walk} $W$, is then
\[
C(W) := \max_{v\in V(G)} C(W,v).
\]
Therefore, the cost of a robot walk on a graph is the maximum weighted latency over all vertices in the graph.  This corresponds to the maximum importance-weighted time between observations of any region.  The min-max weighted latency walk problem can be stated as follows.

\vskip1em
\noindent \textbf{The min-max weighted latency walk problem:}  Find an infinite walk $W$ that minimizes the cost $C(W)$.\\

\noindent For brevity, we will refer to this problem as the \emph{min-max latency walk problem} in the rest of the paper.  

\subsection{Well-Posedness of the Problem}
\label{sec:well-posedness}

Finding an infinite walk is computationally infeasible. Instead, we will try to find the kernel of the minimum cost infinite walk. The first question, however, is whether there always exists a minimum cost walk.
\begin{lemma}
\label{lem:existsOPT}
For any graph $G$, there exists a walk of minimum cost.
\end{lemma}
\begin{proof}
Suppose that there does not exist a walk of minimum cost in $G$.  There are two ways in which this could happen.  Either every walk in $G$ has infinite cost, or there exists an infinite sequence of walks $W_1,W_2,\ldots$ with costs $C(W_1) \geq C(W_2) \geq \cdots$, such that $\lim_{i\to\infty} C(W_i) = c^*$, but there is no walk in $G$ attaining the cost $c^*$.  Thus, to prove the result we will eliminate each of these cases.

First, let $W$ be any walk of size $n$ in $G$ that visits all vertices in $V(G)$. Then the cost $C(\Delta(W))$ is necessarily finite.  Next, we show that there are only a finite number of different values $c'<C(\Delta(W))$ that can be costs of walks in $G$.  Since the length of each edge is positive, for any vertex $v\in V(G)$ there is a finite number of walks beginning in $v$ with length less than {$C(\Delta(W))/\phi(v)$}. Therefore, there are a finite number of possible values for the latency of $v$ that are less than {$C(\Delta(W))/\phi(v)$}. Hence, there is a finite number of possible costs for vertex $v$ on a walk of cost less than $C(\Delta(W))$. Moreover, since there are $n$ vertices in $G$, a walk in $G$ can only have a finite number of different costs. As a result, there exists a walk of minimum cost for graph~$G$.
\end{proof}

We define $\opt_G$ to be the minimum cost among all infinite walks on $G$. An infinite walk $W$ is an optimal walk if $C(W)=\opt_G$. By Lemma~\ref{lem:existsOPT}, such a number always exists. The next question is whether there always exists a finite size kernel realizing $\opt_G$, that is, whether there is an optimal walk that consists of an infinite number of repetitions of a finite size walk.

\begin{lemma}
\label{lem:existsfiniteOPT}
For any graph $G$, there is walk of minimum cost that has a finite size kernel.
\end{lemma}
\begin{proof}
Assume $W$ is an optimal walk with cost $\opt_G$. Note that $W$ is an infinite walk. Let $v$ be a vertex in $V(G)$. Let $W_1$ be a sub-walk of $W$ starting at $v$ with length larger than $\opt_G\times\rho_G$. Since $l(W_1)>\opt_G\times\rho_G$, every vertex of $G$ is visited at least once in $W_1$; otherwise $C(W)\ge l(W_1)/\rho_G >\opt_G$.  Let $U$ be the set of all possible walks in $G$ starting at $v$ with lengths between $l(W_1)$ and $l(W_1)+\max_{e\in E(G)}l(e)$. Since the edge lengths are positive and finite, the size of $U$ is finite.

Let $i$ be the index of a visit to $v$ in $W$.  Then, for every walk $W(i,j)$ of length at least $l(W_1)$, there exists a sub-walk starting at $W(i,i)$ that is in $U$.
Due to the fact that $W$ is an infinite walk and hence $v$ appears an infinite number of times in $W$, there exists a walk $W'\in U$ that appears at least twice in $W$. Let $W''$ be such that $[W',W'',W']$ is a sub-walk of~$W$. Note that $[W',W'']$ is a finite walk, and so $\Delta([W', W''])$ is a walk with a finite size kernel.  We now claim that $\Delta([W', W''])$ is also an optimal walk for $G$.  Consider any two consecutive instances of a vertex $u\in V(G)$ in $\Delta([W', W''])$.  For these two instances of $u$, one of the following two cases occurs:  (i) the two instances are in the same copy of $[W', W'']$, or (ii) the two instances are in consecutive copies of $[W', W'']$.  However, since $[W',W'',W']$ is a sub-walk of $W$, both cases occur in optimal walk $W$ and thus $\Delta([W', W''])$ is also an optimal walk for $G$.
\end{proof}

Next we will show that the problem of min-max latency walk is APX-hard, implying that there is no polynomial-time approximation scheme (PTAS) for it, unless P=NP. However, before that we need to introduce the following definitions.

A \emph{complete graph} is a graph that each pair of its vertices are connected by an edge. A graph is called a \emph{metric graph} if (i) it is a complete undirected graph, and (ii) for any three vertices $u,v,w\in V(G)$ we have $l(uw)\le l(uv)+l(vw)$ (\emph{triangle inequality}) \cite{VVV:'04}.

\begin{theorem}
\label{the:tsp-hardness}
The min-max latency walk problem is APX-hard.
\end{theorem}
\begin{proof}
The reduction is from the \emph{metric Traveling Salesman Problem (TSP)}. Recall that the TSP is the problem of finding the shortest closed walk that visits all vertices exactly once (except for its beginning vertex). Such a walk is referred to as a TSP tour. The problem of finding TSP tours in metric graphs is called the metric TSP. It is known that the metric TSP is APX-hard \cite{PC-YM:93}, and it is approximable within a factor of 1.5. Here we show a reduction from the metric TSP to the min-max latency walk problem that preserves the hardness of approximation.

Let $G$ be the input of the metric TSP. Assign weight $1$ to all vertices of $G$. Assume $W$ is an infinite walk with optimal cost $\opt_G$ in $G$. Let $M$ be a closed walk that is an optimal solution for TSP in $G$ with $l(M)=c'$. We prove {$c'=\opt_G$}. Since each vertex is visited exactly once in $M$, the cost of $\Delta(M)$ is $c'$. However, due to optimality of $W$ we have $\opt_G \le c'$. It remains to show that $c' \le \opt_G$.

Let $v\in V(G)$ be a vertex with  {$C(W,v)=\opt_G$} and $i$ and $j$ be the indices of two consecutive instances of $v$ with  {$l\big(W(i,j)\big)=\opt_G$}. Since the weight of every vertex is $1$ and $l\big(W(i,j)\big)=\opt_G$, all vertices of $G$ appear in $W(i,j)$; otherwise $C(W)>\opt_G$. Consider the spanning tour $T$ that is obtained from $W(i,j)$ by removing all but one of the instances of each vertex. Since we only remove vertices and due to the triangle inequality we have $l(T)\le l\big(W(i,j)\big)$. Since $T$ visits each vertex of $G$ exactly once, it is a candidate solution for TSP and hence we have $l(M)\le l(T)$. Therefore, $c' = l(M)\le l\big(W(i,j)\big) = \opt_G$. Note that we showed that the size of the solution for the two problems are equal, hence the reduction is gap preserving and the APX-hardness carries over.
\end{proof}

For a graph $G$ and two vertices $u,v\in V(G)$, the shortest-path distance between $u$ and $v$ is denoted by $d(u,v)$. We focus on solving the min-max latency walk problem only for metric graphs. The reason is that for any graph $G$ and any $u,v\in V(G)$ we can create a graph $G'$ with the same set of vertices such that edge $uv$ in $G'$ has length equal to the shortest-path distance from $u$ to $v$ in $G$, i.e., $l(uv)=d(u,v)$. Then, we construct a walk for $G$ based on a walk in $G'$ by replacing each edge $uv$ with the shortest path connecting $u$ and $v$ in $G$.  {Since $\opt_G=\opt_{G'}$ and any walk in $G'$ corresponds to a walk of lower or equal cost in $G$, any approximation in $G'$ carries over to $G$}. In the literature, the graph $G'$ is refereed to as the \emph{metric closure} of $G$ \cite{VVV:'04}. It should be noted that to aid the presentation in some examples we show non-complete graphs with the understanding that we are referring to their metric closures.

In the proof of Theorem~\ref{the:tsp-hardness} gave a reduction from the TSP to the min-max latency walk problem.  However, in general the TSP tour is not a good approximation for the min-max latency walk problem.
\begin{lemma}
\label{lem:TSPbad}
The cost of a TSP tour of $G$ can be larger than $(n-1)\opt_G$.
\end{lemma}
\begin{proof}
Let $G$ be a graph of $n$ vertices constructed as follows:
\begin{itemize}
\item There is a vertex $v$ with weight $1$.
\item There are $n-1$ vertices, $v_1,v_2,\ldots,v_{n-1}$, each having weight $1/n$.
\item There exists an edge connecting $v_i$ to $v$  with length $1$ for any $1\leq i < n$.
\item There exists an edge connecting $v_i$ to $v_{i+1}$ with length $2$ for any $1\leq i < n-1$ (see Figure \ref{fig:tspBad}).
\end{itemize}
It is easy to see that the triangle inequality holds for edge lengths in $G$. The TSP tour has length $2n-2$ and hence the cost of the TSP tour is $2n-2$. However, the cost of the walk that only uses edges of unit weight and visits $v$ at every other index would be $2$. This means that the cost of TSP can be as bad as $(n-1)$ times the cost of the optimal walk.
\begin{figure}[t]

\centering
\includegraphics[scale=.87]{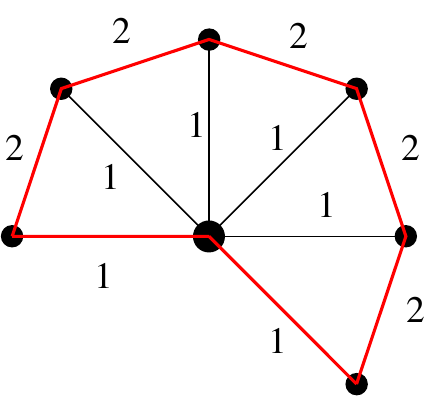}
%
%
\caption{The graph $G$ as in Lemma \ref{lem:TSPbad} with $n=7$. The cost of the TSP tour (red thick edges) in this graph is $2n-2=12$. Note that if the weight of all vertices except the middle vertex is small, there is a walk that has cost $2$.}
\label{fig:tspBad}       
\end{figure}
\end{proof}

In the following sections we seek better approximation algorithms for the min-max latency walk problem.

\section{Relaxations and Simple Bounds}
\label{sec:relaxation}

In this section, we present a relaxation of the min-max latency walk problem and two simple bounds based on the lengths of the edges of the input graph.

\subsection{Relaxation of Vertex Weights}

Here, we define a relaxation of the problem so that all weights are of the form $1/2^x$, where $x$ is an integer.  A similar relaxation has been used in both~\cite{gortz2011capacitated} and~\cite{SLS-DR:09a}.
\begin{definition}[Weight Relaxation]
We say weights of vertices of graph $G$ are \emph{relaxed}, if for any vertex  {$v\in V(G)$} we update its weight $\phi(v)$ to {$\phi'(v)=\frac{1}{2^x}$} such that $x$ is the smallest integer for which {$\frac{1}{2^x} \leq \phi(v)$} holds.
\end{definition}
\begin{lemma}[Relaxed Vertex Weights]
\label{lem:relaxing}
For a graph $G'$ obtained by relaxing the weights of vertices of $G$ the following statements hold:
\begin{enumerate}
  \item If a walk $W$ has cost $c$ in $G$ and cost $c'$ in $G'$, then $c' \leq c < 2c'$.
  \item $\opt_{G'} \leq \opt_G < 2\opt_{G'}$.
\end{enumerate}
\end{lemma}
\begin{proof}
(i) The weight of each vertex in $G'$ is less than or equal to the weight of that vertex in $G$, while the lengths of the corresponding edges are the same. Hence, for costs of $W$ in $G$ and $G'$ we have $c' \leq c$. Moreover, the weight of each vertex in $G'$ is more than half of the weight of the same vertex in $G$. This results in $c < 2c'$. Consequently, we have $c' \leq c < 2c'$.

(ii) Let $W$ and $W'$ be optimal walks in $G$ and $G'$, respectively. For cost of $W$ in $G'$, denoted by $c$, we have $\opt_{G'}\le c$. Also, (i) results in $c\le \opt_G<2c$. Consequently, we have $\opt_{G'}\le \opt_G$. Similarly, for cost of $W'$ in $G$, denoted by $c'$, we have $\opt_G\le c'$. Moreover, by (i) it follows that $\opt_{G'}\le c'<2\opt_{G'}$ and hence $\opt_G<2\opt_{G'}$. Therefore, we have $\opt_{G'} \leq \opt_G < 2\opt_{G'}$.
\end{proof}

The reason for considering this relaxation is as follows.  Given a relaxed graph $G'$, we can define $V_i$ to be all vertices in $G'$ with weight $1/2^i$.  Then, in order for the vertices in $V_i$ and $V_{i+1}$ to have the same weighted latency, each vertex in $V_i$ should be visited twice as often as each vertex in $V_{i+1}$ in a walk on $G'$. This observation gives us some structure that we can exploit in our search for approximation algorithms for walks on $G'$.  By Lemma~\ref{lem:relaxing}(ii), an $\alpha$-approximation algorithm on $G'$ would yield a  $2\alpha$-approximation algorithm on the unrelaxed graph $G$.

\subsection{Simple Bounds on Optimal Cost}

It is easy to observe that no vertex can be too far away from a vertex with weight one, as this distance will bound the cost of the optimal solution.
\begin{lemma}
\label{lem:maxdishalfc}
Let $G$ be a metric graph. For any vertices $u,v\in V(G)$ such that $v$ has weight $1$, we have $l(uv) \leq \opt_G/2$.
\end{lemma}
\begin{proof}
By way of contradiction, assume that $l(uv)>\opt_G/2$ for some $u,v\in V(G)$ such that {$\phi(v)=1$}. Let $W$ be an optimal walk in $G$, i.e., $C(W)=\opt_G$. Let $u_i$  be an occurrence of $u$ in $W$. Let $v_j$ and $v_k$ be the two consecutive occurrences of $v$ preceding and succeeding $u_i$ in $W$, respectively. Since $G$ is metric, the sub-walk of $W$ that lies between $v_j$ and $v_k$ has length greater than $\opt_G$. However, since  {$\phi(v)=1$}, this contradicts the assumption that $W$ has cost $\opt_G$.
\end{proof}

\begin{corollary}
\label{cor:maxdisc}
If $G$ is a metric graph, then the maximum edge length in $G$ is at most $\opt_G$.
\end{corollary}
\begin{proof}
The case that one end of an edge $e$ has weight 1 is addressed in Lemma \ref{lem:maxdishalfc}. Therefore, we consider an edge $uv$ such that $\phi(u),\phi(v)<1$. Let $w$ be a vertex in $V(G)$ such that $\phi(W)=1$. By Lemma \ref{lem:maxdishalfc} it follows that $l(uw),l(vw)\le \opt_G/2$. Moreover, since $G$ is a metric graph, we have $l(uv)\le l(uw)+l(vw)\le \opt_G$.
\end{proof}

\section{Properties of Min-Max Latency Walks}
\label{sec:existence_results}

In this section, we characterize the optimal and approximate solutions of the min-max latency walk problem.

\subsection{Bounds on Size of Kernel of an Optimal Walk}
\label{subsec:kernel-bounds}
Here, we show that the optimal solution for the min-max latency walk problem can be very large with respect to the size of the input graph.
\begin{lemma}
\label{lem:nonPoly}
There are infinitely many graphs for which any optimal walk has a kernel that is at least exponential in the size of $G$.
\end{lemma}
\begin{proof}
For any constant integer $k$ and any integer multiple of it $n=sk$, we construct a graph $G$ with unit length edges and $|V(G)|=n$ and prove that the smallest kernel of any optimal solution has size in $\Omega(n^{k-1})$. Let $V_1,\ldots,V_k$ be a partition of $V(G)$ into $k$ sets each having size $s$. Let there be a unit length edge $uv$ for any $u\in V_1$ and $v\in V_i$, where $i \in \{1,2,\ldots,k\}$. For each $v\in V_i$ with $1\leq i \leq k$, let {$\phi(v)=\frac{1}{{(s+1)}^i}$}. We first prove that $\opt_G \leq 1$.

Let $W$ be a walk constructed by visiting all vertices in the sets $V_1,V_2,\ldots,V_{i-1}$ recursively between any two consecutive visits to members of $V_i$ (see Algorithm \ref{alg:constpath}). It is easy to see that cost of $\Delta(W)$ is at most $1$. The reason is that each vertex in $V_i$  {for $i \in \{1,2,\ldots,k\}$} has weight $\frac{1}{(s+1)^i}$ and is visited in $\Delta(W)$ at least once every other $(s+1)^i$ steps by the construction (see Figure \ref{fig:walkKS}). Therefore {$C(\Delta(W),v)$} is bounded by $1$ for any vertex $v$.

We have proved $\opt_G \leq 1$. It remains to prove any infinite walk $M$ in $G$ with cost less than or equal to $1$ has a kernel of size $\Omega(n^{k-1})$. Let $M_1$ be a sub-walk of size $s+1$ of $M$. Then all vertices of $V_1$ appear in $M_1$, otherwise the vertex $v$ in $V_1$ that does not appear in $M_1$ would induce a cost larger than $1$ to $M$, that is, {$C(M,v)\geq(s+2)\times\frac{1}{s+1}>1$}. This means that after each visit to a member of $V_i$ with $i>1$, the next $s$ vertices that are visited in $M$ all belong to $V_1$.

Now we need to show that  {at most} a single vertex in $\bigcup_{j>i}V_{j}$ appears in any sub-walk of $M$ of size $(s+1)^{i-1}$. To prove this we use induction on $i$. Let $M'$ be a sub-walk of $M$ with size $(s+1)^{i-1}$. We can partition the elements of $M'$ into $s+1$ disjoint sub-walks of size $(s+1)^{i-2}$. By the induction hypothesis, we know that each part of this partition has {at most} a single instance of vertices in $\bigcup_{j> i-1}V_j$. Also, we know that all vertices of $V_i$ appear in $M'$, or else the vertex $v\in V_i$ that is not visited in $M'$ would have cost {$C(M,v)>1$}. Since there are $s$ vertices in $V_i$ and $s+1$ visits to  vertices of $\bigcup_{j>i-1}V_j$ in $M'$, there is {at most} a single visit to a vertex in $\bigcup_{j>i}V_j$ in $M'$. Since all vertices in $V_k$ appear  in the kernel of $M$, this means that the kernel of $M$ has size at least $(s+1)^{k-1}$. Since $k$ is a constant and $n=sk$, therefore $(s+1)^{k-1}\in\Omega(n^{k-1})$. Hence, kernel of any optimal walk is at least exponential in the size of $G$.
\end{proof}

\begin{figure}[t]
\centering
\includegraphics[scale=.9]{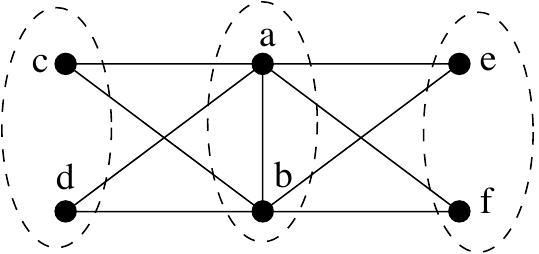}
\caption{The graph $G$ as in proof of Lemma \ref{lem:nonPoly} with $n=6$, $s=2$, $V_1=\{a,b\}$, $V_2=\{c,d\}$ and $V_3=\{e,f\}$. The walk that Algorithm \ref{alg:constpath} constructs would be $[[[a,b],c,[a,b],d],[a,b],e,[[a,b],c,[a,b],d],[a,b],f]$, where brackets show recursive calls in Algorithm \ref{alg:constpath}.}
\label{fig:walkKS}       
\end{figure}

\begin{algorithm}
\caption{$\textsc{WalkMaker}(\{V_1,\ldots,V_{i-1},V_i\})$}
\begin{algorithmic}[1]
\IF{$i<1$}
    \RETURN $\emptyset$
\ELSE
    \STATE $W \leftarrow \emptyset$
    \FOR{$j = 1 \to |V_i|$}
        \STATE $W  \leftarrow  [W,\textsc{WalkMaker}(\{V_1,\ldots,V_{i-1}\})]$,
        \STATE $W \leftarrow [W,v]$; where $v$ is the $j$-th element in $V_i$
    \ENDFOR
    \RETURN $W$
\ENDIF
\end{algorithmic}
\label{alg:constpath}
\end{algorithm}

\begin{corollary}
\label{cor:nopolyalg}
There does not exist a polynomial time algorithm for the min-max latency walk problem.
\end{corollary}
Corollary \ref{cor:nopolyalg} does not show exactly how hard the problem is. In fact, any algorithm that checks all possible walks to find the optimal solution will have at least doubly exponential time complexity.

\subsection{Binary Walks}

In Section \ref{subsec:kernel-bounds}, we showed that any exact algorithm is not  scalable with respect  to the size of the input graph. Therefore, we turn our attention to finding walks that approximate the optimal cost of the graph. We show there always exists a polynomial size walk that has a cost within a constant factor of the optimal cost. To obtain this result, we first need to define a special class of walks and show that there are walks in this class that provide constant factor approximations.

\begin{definition}[Binary Walks and Decompositions]
\label{def:binary-walk-decomp}
Let  {$G'$} be a relaxed graph and $V_i$ be the set of vertices with weight $1/2^{i}$ in {$G'$}. A walk $S$ is a \emph{binary walk} if it can be written as $[S_1,S_2,\ldots,S_t]$, where  {$t = 2 \rho_{G'}$}, such that for any $v\in V_i$ and any {$0\leq j < t/2^i$}, vertex $v$ appears exactly once in $[S_{j2^i+1},S_{j2^i+2},\ldots,S_{(j+1)2^i}]$. In other words, in each $2^i$ consecutive $S_l$'s starting from $S_{j2^i+1}$, vertex $v$ appears exactly once. Also, we say that the tuple of walks $(S_1,S_2,\ldots,S_t)$ is a \emph{binary decomposition} of $S$.
\end{definition}
By Definition \ref{def:binary-walk-decomp}, each vertex appears in each $S_l$ at most once. Therefore, the size of each $S_l$ is bounded by $n$, where $n=|V(G')|$. This means that $S$ has size bounded by {$2n\rho_{G'}$}. Consider a binary walk $S=[S_1,\ldots,S_t]$, its expansion $\Delta(S)$, and a member of the binary walk $S_l=(v_1,\ldots,v_k)$ for some $1\leq l \leq t$.  We say that a sub-walk $\Delta(S)(i,j)$ of $\Delta(S)$ \emph{intersects} $S_l$ if either $(v_1,\ldots,v_{i'})$ or $(v_{i'},\ldots,v_{k})$ appears in $\Delta(S)(i,j)$ for some $1\leq i'\leq k$.



\begin{lemma}
\label{lem:atmost_2^{i+1}}
Let $G'$ be a relaxed graph and let $V_i$ denote the set of all vertices of $G'$ with weight $1/2^{i}$. Let $S=[S_1,S_2,\ldots,S_t]$ be a binary walk in $G'$. If $a$ and $b$ are indices of two consecutive visits to a vertex $v\in V_i$ in $\Delta(S)$, then $\Delta(S)(a,b)$ intersects at most $2^{i+1}$ members of {$(S_1,S_2,\ldots,S_t)$}.
\end{lemma}

\begin{proof}
Since $\Delta(S)$ is constructed by the concatenation of an infinite number of copies of walk $S$, the two following cases can arise for indices $a$ and $b$: (i) indices $a$ and $b$ refer to two visits of vertices in the same copy of $S$, or (ii) indices $a$ and $b$ refer to two visits such that they are in two consecutive copies of $S$. Note that since by Definition \ref{def:binary-walk-decomp} every vertex is visited at least once in a binary walk $S$, no other case is possible. For case (i), by Definition \ref{def:binary-walk-decomp} we have that $\Delta(S)(a,b)$ intersects $2^i+1$ members of {$(S_1,S_2,\ldots,S_t)$}. For case (ii), let $a'\ge a$ be the maximum index of a vertex in $\Delta(S)$ such that visits $a$ and $a'$ are in the same copy of $S$. Similarly, let $b'\le b$ be the minimum index of a vertex in $\Delta(S)$ such that visits $b'$ and $b$ are in the same copy of $S$. Since there is one visit to $v$ in both $\Delta(S)(a,a')$ and $\Delta(S)(b',b)$, each of these two sub-walks intersects at most $2^i$ members of $(S_1,S_2,\ldots,S_t)$. Consequently, $\Delta(S)(a,b)$ intersects at most $2^{i+1}$ members of $(S_1,S_2,\ldots,S_t)$.
\end{proof}

\begin{lemma}
\label{lem:chopchop}
Let {$G'$} be a graph with relaxed weights. There is a binary walk $S$ in  {$G'$} with cost at most {$2.5\times\opt_{G'}$} and size bounded by {$2n\rho_{G'}$}, where $n=|V(G')|$.
\end{lemma}
\begin{proof}
Let $M=(m_1,m_2,\ldots)$ be an optimal infinite walk in $G'$ with cost $c=\opt_{G'}$. Since $M$ is an infinite walk, we can begin the walk at any vertex $m_i$ and obtain the cost $c$.  Therefore, we assume without loss of generality that $m_1$ is such that {$\phi(m_1)=1$}. Based on $M$, we construct a binary walk $S$ such that the cost of $\Delta(S)$ is at most $2.5 c$ as follows. Let $a_0$ be $0$ and $S_i$ be the sub-walk $M(a_{i}+1,a_{i+1})$ such that $a_{i+1}$ is the maximal index satisfying {$l\big(M(1,a_{i+1})\big)\leq ic$}. Therefore, we have $l\big(M(1,a_i+1)\big)> (i-1)c$ and hence $l\big(M(a_i+1,a_{i+1})\big)\le c$. Consequently, each $S_i$ is a walk of length at most $c$ such that the union of $S_i$'s partitions $M$.

Now we modify the walks $S_1,S_2,\ldots$ by omitting some of the instances of vertices in them. Let $V_i$ be the set of vertices with weight $1/2^{i}$ in  {$G'$}. Let  {$t = 2 \rho_{G'}$} as in Definition \ref{def:binary-walk-decomp}. For every vertex $u\in V_i$ and any number $0\leq j <t/2^i$, omit all but one of the instances of $u$ that appear in  {$S_{j2^i+1},S_{j2^i+2},\ldots,S_{(j+1)2^i}$}. There exists at least one such instance; otherwise a vertex $u$ with weight $1/2^i$ exists that is not visited in an interval of length greater than $c\times2^i$, implying  {$C(M,u) > c$}. Let $S_1',S_2',\ldots$ be the result of this modification, note that {$l(S_i')\le l(S_i)$} for each $1\le i\le t$.

Let $S$ be $[S_1',S_2',\ldots, S_t']$. We claim that $\Delta(S)$ has cost at most $2.5 c$. For $u\in V_i$ a vertex of  {$G'$}, we know that $u$ appears exactly once in each $2^i$ consecutive $S'_l$'s, i.e., {$[S'_{j2^i+1},S'_{j2^i+1},\ldots,S'_{(j+1)2^i}]$} for any  {$0\le j<t/2^i$}. Therefore by Definition \ref{def:binary-walk-decomp}, $S$ is a binary walk. By Lemma \ref{lem:atmost_2^{i+1}}, two consecutive visits to a vertex $u\in V_i$ in $\Delta(S)$ intersects at most $2^{i+1}$ members of {$(S_1,S_2,\ldots,S_t)$}.

By the construction, we have that for any $j,k$ with  {$0< j \leq k$}, walk $[S_j',S_{j+1}',\ldots,S_k']$ has length at most $(k-j+1) c$. Also since  {$\phi(m_1)=1$}, by Lemma \ref{lem:maxdishalfc} we know that for any $j,k$ with $0 \leq k \leq j$, walk $[S_j',S_{j+1}',\ldots S_t',S_1',S_2',\ldots,S_k']$ has length at most $((t-j+1)+0.5+k) c$. Therefore, if $a$ and $b$ are indices of two consecutive visits to $u\in V_i$ in $\Delta(S)$, then
\[
l\big(\Delta(S)(a,b)\big) < 2^{i+1}c+0.5c \leq 2^i \times (2.5 c).
\]
Consequently, each vertex $u\in V(G')$ has cost {$C(\Delta(S),u)\le 2.5c$} and hence $C\big(\Delta(S)\big)\le 2.5c$. Moreover, since $S$ is a binary walk, its size is bounded by {$2n\rho_{G'}$}.
\end{proof}

\begin{theorem}
\label{the:chopchop}
In any graph $G$ with $n=|V(G)|$, there exists a walk $W$ of size $O(n^2)$ such that the cost of $\Delta(W)$ is less than or equal to  {$6\times \opt_G$.}
\end{theorem}
\begin{proof}
Let $G'$ be the relaxation of $G$ and let $V_i$ be the set of vertices  {$u \in V(G')$} of weight $1/2^i$. The set of vertices in $V(G')$ with weights less than $1/2^{\lfloor \log n \rfloor +1}$ is denoted by $U=\{u_1,u_2,\ldots,u_{|U|}\}$. Let $G''$ be the graph obtained by removing vertices in $U$ from $G'$. Therefore, $\rho_{G''} \leq 2^{\lfloor \log n \rfloor +1} \leq 2n$. Note that $G''$ is also a metric graph. Let $S$ be a binary walk in $G''$ with cost at most $2.5\opt_{G''}$ as described in Lemma \ref{lem:chopchop}. Since $\rho_{G''} \leq 2n$, the size of $S$ is bounded by $2\rho_{G''}n \leq 4n^2$.

Now, we add the vertices in $U$ to $S$ in order to obtain a walk $W$ that covers all vertices of $G'$. Let $v\in V(G')$ be a vertex with  {$\phi(v)=1$}. Let {$(S_1,S_2,\ldots,S_t)$} be the binary decomposition of $S$, where $t=2^{\lfloor \log n \rfloor +2}$. Let $v_i$ be the $i$-th instance of $v$ in $S$. Note that $t>2n$, thus $v$ appears at least $2n$ times in $S$ (see Definition \ref{def:binary-walk-decomp} for $i=0$). For each $1 \leq i \leq |U|$, modify $S$ by duplicating $v_{2i}$ and inserting an instance of $u_i$ between the two copies of $v_{2i}$. Since $|U|<n$, this operation is possible. Let $W$ be the resulting walk. Note that size of $W$ is in $O(n^2)$. We claim that the cost of $\Delta(W)$ is at most $2\opt_{G'}+\opt_G$.

Let $w\in V_i$ be a vertex in $G'$ such that {$\phi(w)\ge1/2^{\lfloor \log n \rfloor +1}$}. Let $a$ and $b$ be the indices of two consecutive visits to $w$ in $\Delta(S)$. Let $a'$ and $b'$ be the indices of the corresponding visits to $w$ in $\Delta(W)$. By Lemma \ref{lem:atmost_2^{i+1}}, sub-walk $\Delta(S)(a,b)$ intersects at most $2^{i+1}$ members of  {$(S_1,S_2,\ldots,S_t)$}. At least half of these $2^{i+1}$ walks in $\Delta(S)$ are the same in $\Delta(W)$, i.e., at least half of these walks have not been altered by duplication of unit weight vertices or insertion of members of $U$ during the construction of $W$ from $S$. Therefore, at most $2^i$ vertices of $U$ lie between indices $a'$ and $b'$ of $\Delta(W)$. 

Furthermore, we inserted the visits to the vertices of $U$ at visits to $v$ with  {$\phi(v)=1$}. Therefore, by Lemma \ref{lem:maxdishalfc}, each of these new detours made to visit a member of $U$ has length at most $2(\opt_G/2)=\opt_G$. Also, by Lemma \ref{lem:chopchop}, we already know that {$C(\Delta(S),u)\leq 2.5\opt_{G'}$}. Hence, we have
\begin{equation}
 {C(\Delta(W),u)< 2.5\opt_{G'}+\opt_G. \label{equ:naive}}
\end{equation}
Note that the extra $0.5$ factor in Lemma \ref{lem:chopchop} is due to the distance of the last vertex of $S_t$ to the first vertex of $S_1$. However, this extra cost can be treated as one of the detours to vertices of $U$, as we avoided adding one of these detours to $S_1$ and $S_t$. This means that we have already accounted for this extra cost in the second part of the righthand side of inequality \ref{equ:naive}.  {Consequently, we have $C(\Delta(W),w)< 2\opt_{G'}+\opt_G$. By Lemma \ref{lem:relaxing}(ii), we have $\opt_{G'}\leq \opt_G$, therefore $C(\Delta(W),w)< 3\opt_G$. Also, by Lemma \ref{lem:relaxing}(i), the cost of $\Delta(W)$ in $G$ would be less than $6\opt_G$. Consequently, we have $C\big(\Delta(W)\big)<6\opt_G$.}
\end{proof}

In the following, we show that Theorem \ref{the:chopchop} is almost tight with respect to the size of the output.
\begin{lemma}
\label{lem:smalloutputfail}
Any algorithm for the min-max latency walk problem with guaranteed output size in $O(n^2/k)$ has approximation factor in $\Omega(k)$.
\end{lemma}
\begin{proof}
Let $\varepsilon$ be a very small positive number and let us construct a graph $G$ as follows.
\begin{itemize}
\item There are $n/2$ vertices of weight $1$, called heavy vertices.
\item There are $n/2$ vertices of weight $\varepsilon$, called light vertices.
\item Any two heavy vertices are connected to each other by an edge of length $\varepsilon$.
\item There is an edge of length $1$ connecting any light vertex to any heavy vertex.
\end{itemize}

Any minimum cost infinite walk in $G$ visits all heavy vertices between visits to any two light vertices. This means that each heavy vertex is repeated $n/2$ times in any walk that expands into a minimum cost infinite walk. Therefore, any optimum solution has size in $\Omega(n^2)$ and its cost is upper-bounded by  $2+\varepsilon\times O(n)$. The value of $\varepsilon$ can be chosen small enough so that the cost of optimal walk is close to 2.

To reduce the size of the output walk by a factor of $k$, we need to visit at least $k$ light vertices between two consecutive visits to a heavy vertex $v$. This means that a walk of size smaller than $\frac{n^2}{4k}$ has cost at least $2k$, which is $k$ times the optimal cost. Therefore, any solution for the min-max latency walk problem in $G$ with size in $O(n^2/k)$ has approximation factor in $\Omega(k)$.
\end{proof}

\begin{figure}[t]
\centering
\includegraphics[scale=.4]{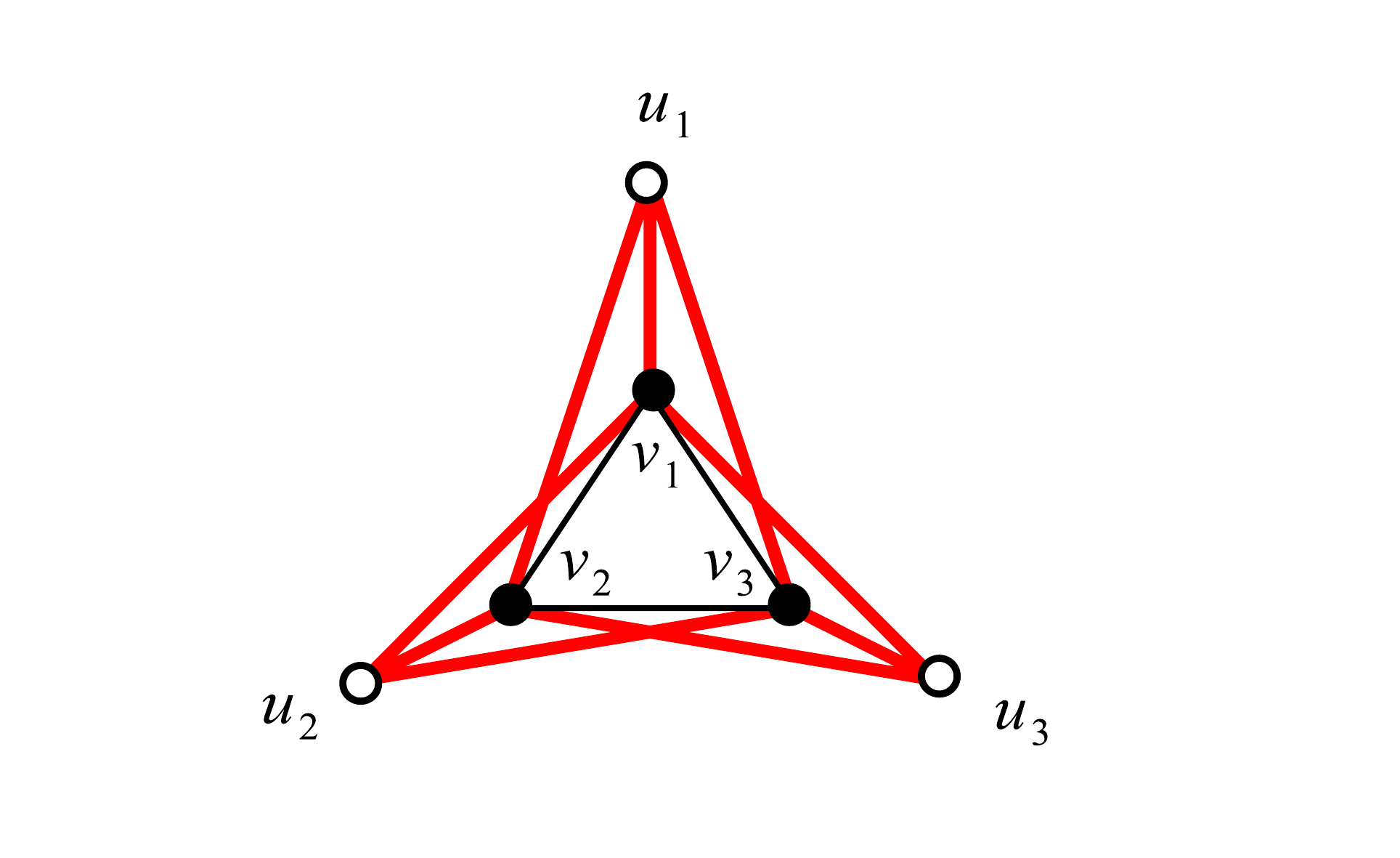}
%
%
\caption{Graph $G$ as in Lemma \ref{lem:smalloutputfail} with $n=6$. Vertices  $v_1,v_2,v_3$ are the heavy vertices and $u_1,u_2,u_3$ are the light vertices. The red thick edges have unit weights and the weight of the black edges is small. We have $C\big(\Delta(W)\big)=2+6\varepsilon$ for $W=(v_1,v_2,v_3,u_1,v_1,v_2,v_3,u_2,v_1,v_2,v_3,u_3)$ with size 12. On the other hand, $C\big(\Delta(W')\big)=4+2\varepsilon$, where $W'=(v_1,u_1,v_2,u_2,v_3,v_1,u_3,v_2,v_3)$ has size 9. Here, we have $k=2$.}
\label{fig:heavy-light}       
\end{figure}

Lemma \ref{lem:smalloutputfail} directly gives that there is no constant factor approximation algorithm with guaranteed output size in $o(n^2)$. Note that this implies that Theorem \ref{the:chopchop} is tight in the sense that the size of the constructed kernel can  be reduced by at most a constant factor.

\section{Approximation Algorithms for the Min-Max Latency Walk Problem}
\label{sec:approx_algs}

In this section, we present two polynomial time approximation algorithms for the min-max latency walk problem. The approximation factor of the first algorithm is a function of the ratio of the maximum weight to the minimum weight among vertices, i.e., $\rho_G$. The approximation ratio of the second algorithm, however, relies solely on the number of vertices in the input graph.

\subsection{An $O(\log \rho_G)$-Approximation Algorithm}

%
A crucial requirement for our algorithms is a useful property regarding binary walks. This property is discussed in the following lemma.

\begin{lemma}(Binary Property)
\label{lem:perfectdecomposition}
Let {$G'$} be a graph with relaxed weights. Let $S$ be a binary walk in  {$G'$} with the binary decomposition  {$(S_1,S_2,\ldots,S_t)$}. Assume that
\begin{enumerate}
\item for some $c\in \mathbb{R_+}$, {$\max_{1\leq i\leq t}l(S_i) \leq c$}, and
\item each $S_i$ begins in a vertex $v\in V(G')$, where $\phi(v)=1$.
\end{enumerate}
Then the cost of $\Delta(S)$ in $G'$ is  {at most $2c+\opt_{G'}$.}
\end{lemma}
\begin{proof}
Let $V_i$ be the set of vertices of weight $1/2^i$ in {$V(G')$}. Let $u\in V_i$ be a vertex of  {$G'$}. By Lemma \ref{lem:atmost_2^{i+1}}, for any $a$ and $b$ that are the indices of two consecutive visits to $u$ in $\Delta(S)$, we have that
$\Delta(S)(a,b)$ intersects at most $2^{i+1}$ members of {$(S_1,S_2,\ldots,S_t)$}. Also, by condition (ii) and Lemma \ref{lem:maxdishalfc} we know that the lengths of edges connecting two consecutive $S_l$'s in $\Delta(S)$ are at most $\opt_{G'}/2$.
Hence, we have {$l\big(\Delta(S)(a,b)\big) \leq 2^{i+1}(c+\opt_{G'}/2)$}. Consequently, each vertex $u\in V(G')$ has cost {$C(\Delta(S),u)\le 2c+\opt_{G'}$}. Hence, we have $C\big(\Delta(S)\big)\le 2c+\opt_{G'}$.
%
\end{proof}

Here, we define a tool that will be useful in our approximation algorithms. Let $\mathrm{Partition}(W, k)$ be a function that takes as input a walk $W$ and an integer $k$ and returns a set of $k$ walks $\{W_1,W_2,\ldots,W_k\}$ that partitions the vertices of $W$ such that {$l(W_i) \leq l(W)/k$}, for all $1\leq i\leq k$. It is easy to see this can always be computed in linear time by a single traversal of $W$.
Note that in case that $k>|W|$, $\mathrm{Partition}(W, k)$ returns a set  $\{W_1,W_2,\ldots,W_k\}$ in which $W_i=W(i,i)$, for $1\le i\le |W|$, and $W_i=\emptyset$ for $|W|<i\leq k$.
Also, recall from Section~\ref{sec:graph_background} that $\text{TSP-Path}(G)$ is a walk that visits each vertex in $G$ exactly once.

\begin{algorithm}
\caption{$\textsc{BrutePartitionAlg}(G)$}
\begin{algorithmic}[1]
    \STATE Let $V_i$ be the set of vertices of weight  {$\frac{1}{2^{i}} \le \phi(u) < \frac{1}{2^{i-1}}$ for $0 \leq i \leq \lceil\log_2 \rho_G\rceil$}
    \STATE Let $t$ be $2^{\lceil \log_2 \rho_G \rceil +1}$
    \STATE  {$S_1,S_2,\ldots,S_t \leftarrow \emptyset$}
    \FOR{$i = 0 \to \lceil \log_2 \rho_G \rceil$}
        \STATE  {$\{W_{i,0},\ldots,W_{i,2^i-1}\} \leftarrow \mathrm{Partition}(\text{TSP-Path}(G[V_i]), 2^i)$}
        \FOR{ {$k = 1 \to t$}}
            \STATE $S_k \leftarrow [S_k,W_{i,j_i}]$; where $j_i$ is $k$ $\mathrm{mod}$ $2^i$,
        \ENDFOR
    \ENDFOR
    \STATE $S\leftarrow[S_1,\ldots,S_t]$
    \RETURN $S$
\end{algorithmic}
\label{alg:const_logarithmic_partition_epsilon}
\end{algorithm}

Given a graph $G$, our first approximation algorithm, shown in Algorithm~\ref{alg:const_logarithmic_partition_epsilon}, is guaranteed to find a solution with cost within a factor of $O(\log\rho_G)$ of the optimal cost, where $\rho_G$ is ratio of the maximum vertex weight to the minimum vertex weight in $G$. The main idea in Algorithm~\ref{alg:const_logarithmic_partition_epsilon} is to construct a binary walk {$S=[S_1,S_2,\ldots,S_t]$} that satisfies the binary property discussed in Lemma \ref{lem:perfectdecomposition}. Recall that in Lemma \ref{lem:perfectdecomposition} it was shown that if {$\max_{1\leq k \leq t}l(S_k)<c$} for some $c$ and each $S_k$ begins in $v$ for some vertex $v\in V(G)$ with $\phi(v)=1$, then cost of $S$ is at most $2c+\opt_{G'}$. Here, we obtain a method that constructs a binary walk $S=[S_1,S_2,\ldots,S_t]$ such that for each $1\le k\le t$, $S_k$ starts with a unit weight vertex $v$ and the length of $S_k$ is in $O(\log \rho_G)$. An overview of Algorithm~\ref{alg:const_logarithmic_partition_epsilon} is as follows. First graph $G$ is relaxed. Then a TSP-Path is calculated for each set of vertices $V_i$, where $V_i$ denotes the set of all vertices with relaxed weight of $1/2^i$. The TSP-Path through vertices in $V_i$ is partitioned into $2^i$ walks. These walks are concatenated to create each $S_k$.  An illustration of Algorithm~\ref{alg:const_logarithmic_partition_epsilon} is given in Figure~\ref{fig:Brute_alg}.  The final walk $S$ is obtained by repeatedly performing the four walks shown in Figure~\ref{fig:1e} through to Figure~\ref{fig:1h}.   

\begin{theorem}
\label{tho:method1}
Given a graph $G$, Algorithm \ref{alg:const_logarithmic_partition_epsilon} constructs a walk $S$ of size in {$O(n\rho_G)$} such that $\Delta(S)$ is an $O(\log \rho_G)$-approximation for the min-max latency walk problem in $G$.
\end{theorem}
\begin{proof}
Let $G'$ be the result of relaxing the weights of $G$. The set of vertices  {$u \in V(G')$} of weight $\frac{1}{2^{i}}$ is denoted by $V_i$. Let $v$ denote a vertex in $V_0$, i.e., $\phi(v)=1$, such that $v$ is the first vertex that appears in $\text{TSP-Path}(G[V_0])$. Let $t$ be the smallest power of two that is larger than $\rho_G$, i.e., $t= \rho_{G'}$. Algorithm \ref{alg:const_logarithmic_partition_epsilon} constructs a binary walk  {$S=[S_1,S_2,\ldots,S_t]$} such that all $S_k$'s begin in $v$ and {$\max_{1\leq k \leq t}l(S_k)<2( \log_2 \rho_{G'})\opt_{G'}$}.

Let us look at set $V_i$ for some $0\le i\le \log_2 \rho_{G'} $. Algorithm \ref{alg:const_logarithmic_partition_epsilon} constructs $2^i$ walks such that each vertex in $V_i$ appears in exactly one of these walks. Let $W_i$ be the output of $\text{TSP-Path}(G[V_i])$. In the proof of Theorem \ref{the:tsp-hardness}, we showed that in graphs with uniform weights the length of the TSP tour is the same as the length of the kernel of the best solution for the min-max latency walk problem. Also, note that in the optimal solution of the min-max latency walk problem in $G'$ the maximum length between two consecutive visits to a vertex in $V_i$ is $2^i\opt_{G'}$; otherwise the cost of the solution is greater than $\opt_{G'}$. Moreover, the length of a TSP tour is an upper-bound for the length of a TSP-Path. Since $W_i$ visits only vertices in $V_i$ and they all have the same weight, we have {$l(W_i)/2^i \leq \opt_{G'}$}.

As in line 5 of Algorithm \ref{alg:const_logarithmic_partition_epsilon}, the output of $\mathrm{Partition}(W_i,2^i)$ is $\{W_{i,0},W_{i,2},\ldots,W_{i,2^i-1}\}$. Hence, we have $l(W_{i,j})\le l(W_i)/2^i \le\opt_{G'}$ for all $0\le j< 2^i$. Observe that the output of $\mathrm{Partition}(W_0,1)$ is $\{W_{0,0}\}$, where $W_{0,0}=W_0$.
A walk $S=[S_1,S_2,\ldots,S_t]$ is constructed using each $W_{i,j}$ as follows.

\begin{figure}
\centering
\includegraphics[scale=0.30]{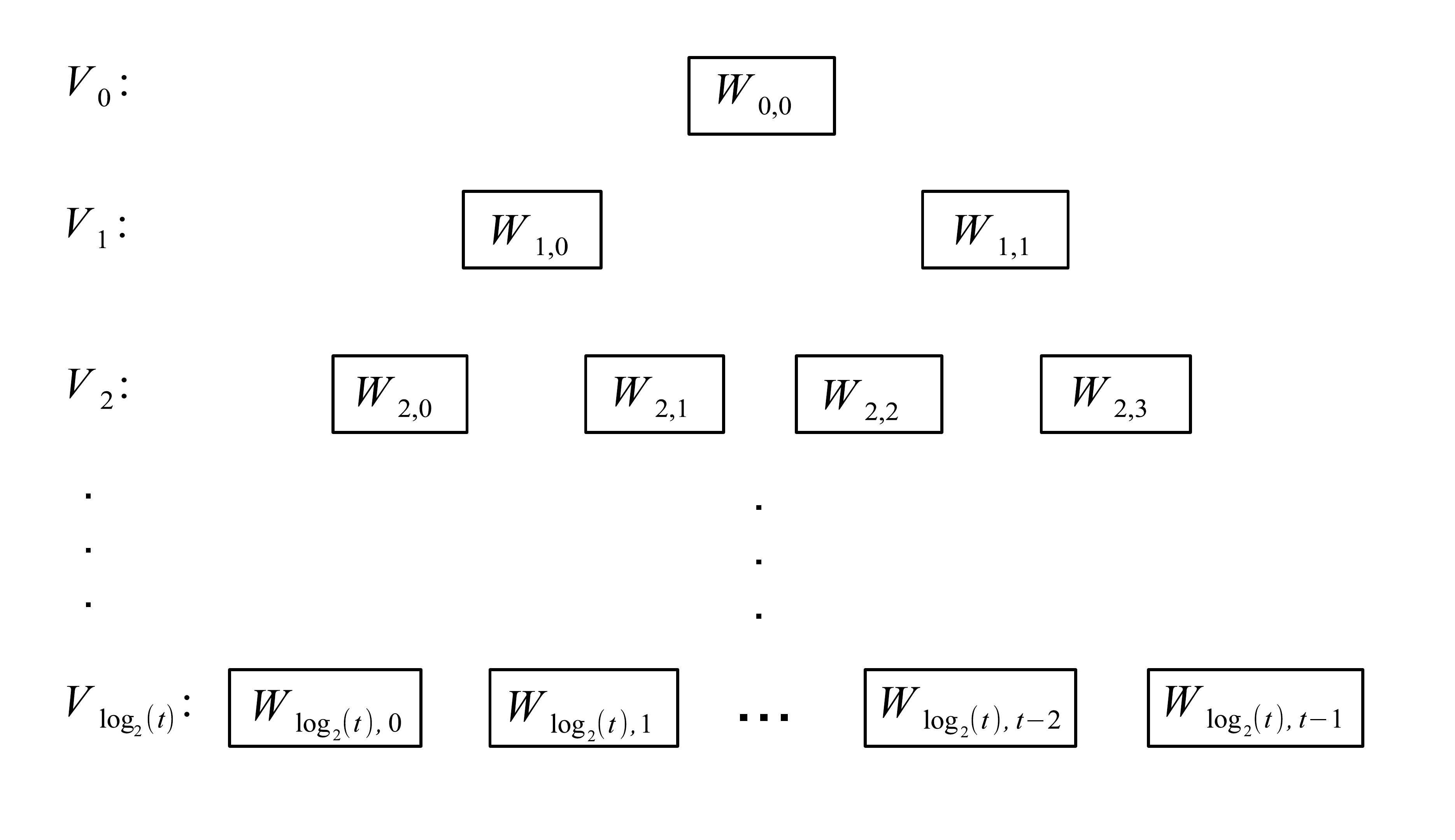}
\caption{The construction of $S_k$'s from walks $W_{i,j}$'s is depicted. For $1\le k\le t$, each $S_k$ is constructed by concatenation of the walks $W_{0,j_0},W_{1,j_1},\ldots,W_{\log_2 t,j_{\log_2 t}}$ while preserving order, where $k\equiv j_i \pmod{2^i}$ for each $0\le i\le \log_2 t$.}
\label{fig:logro}       
\end{figure}

In line 6 of Algorithm \ref{alg:const_logarithmic_partition_epsilon}, for each $1\le k\le t$, walk $S_k$ is constructed by concatenating the walks $W_{0,j_0},W_{1,j_1},\ldots,W_{\log_2 t,j_{\log_2 t}}$ while preserving order, where $k\equiv j_i \pmod{2^i}$ for each $0\le i\le \log_2 t$.
Observe that for $1\le k\le t$, each $S_k$ starts with $W_{0,0}$. Therefore, the first vertex of all $S_k$'s is vertex $v$, where $v$ is the first vertex that appears in $\text{TSP-Path}(G[V_0])$ (see Figure~\ref{fig:logro}).
Note that by Corollary \ref{cor:maxdisc}, the length of edges connecting two $W_{i,j}$'s is at most $\opt_{G'}$.
In the end, there will be $ \log_2 \rho_{G'} $ walks in each $S_k$, for $1\le k\le t$, each of length at most $\opt_{G'}$ connected to each other by edges with length at most $\opt_{G'}$. Therefore  {$\max_{1\leq k\le t}l(S_k) \leq 2 \log_2 \rho_{G'} \opt_{G'}$}.

By Lemma \ref{lem:perfectdecomposition}, walk $S$ has cost at most $4\log_2 \rho_{G'} \opt_{G'}+\opt_{G'}$ in $G'$. Using Lemma \ref{lem:relaxing}(ii), it follows that cost of $S$ in $G'$ is at most $4\log_2 \rho_{G'} \opt_{G}+\opt_{G}$. Hence, by Lemma \ref{lem:relaxing}(i), walk $S$ is an $(8\log_2 \rho_{G'}+2)$-factor approximate solution for $G$. Note that since $\rho_{G'}<2\rho_G$, we have that $8\log_2 \rho_{G'}+2$ is in $O(\log \rho_G)$. Therefore, $\Delta(S)$ is an $O(\log \rho_G)$-approximation for the min-max latency walk problem in $G$. Finally, since each vertex $u\in V_i$ appears exactly once in $W_i$ and by the construction of $S_k$'s, for $1\le k\le t$, we have that $S$ is a binary walk. Therefore, the size of $S$ is in $O(n\rho_G)$.
\end{proof}

\begin{figure}
\begin{subfigure}[b]{.21\linewidth}
\centering
\includegraphics[width=\linewidth]{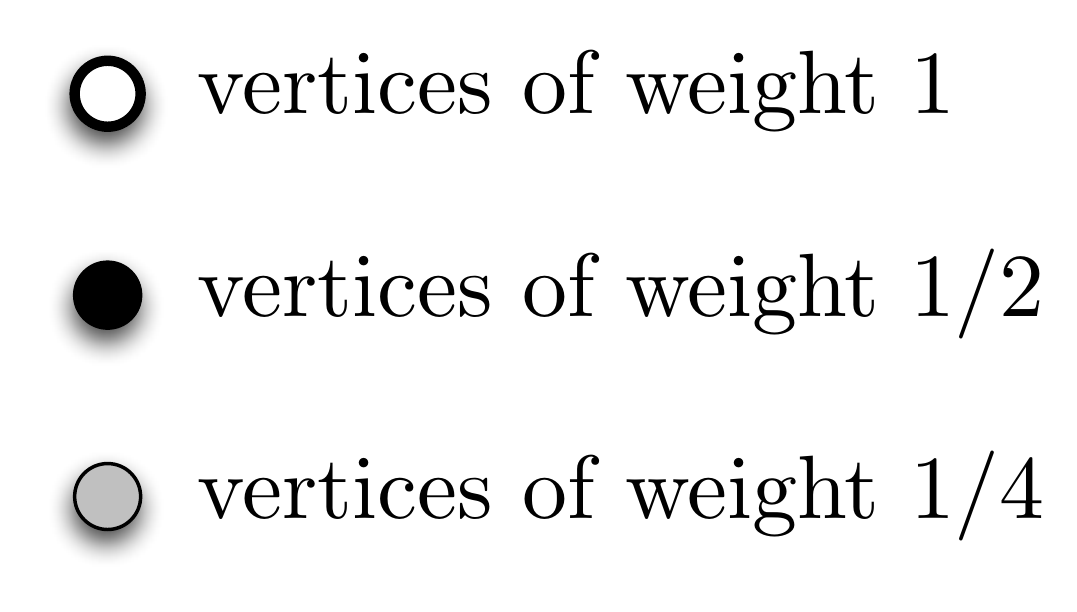}
\caption{Legend for vertices.}
\label{fig:1a}
\end{subfigure}%
\hfill
\begin{subfigure}[b]{.21\linewidth}
\centering
\includegraphics[width=\linewidth]{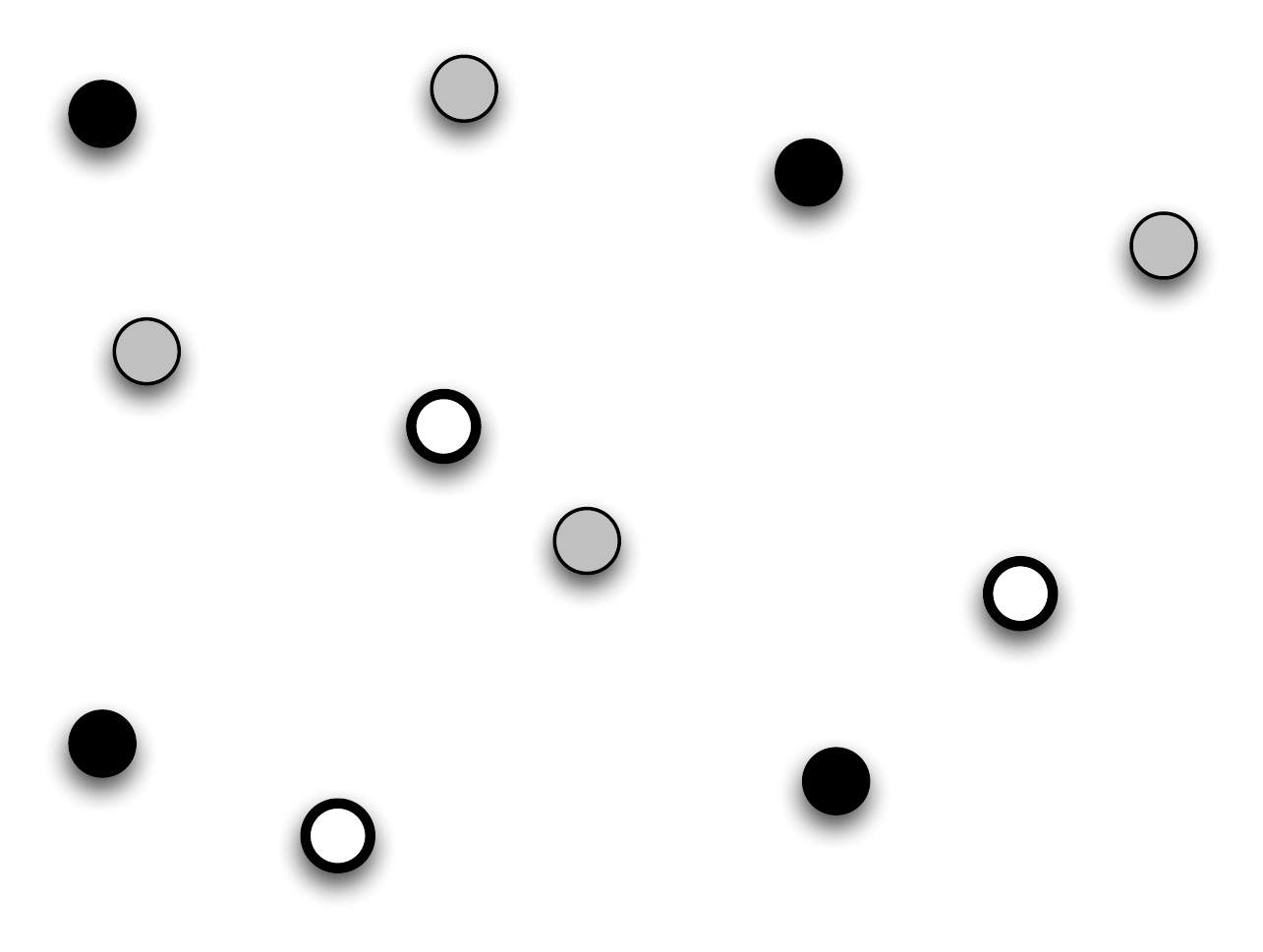}
\caption{A relaxed graph.}
\label{fig:1b}
\end{subfigure}%
\hfill
\begin{subfigure}[b]{.21\linewidth}
\centering
\includegraphics[width=\linewidth]{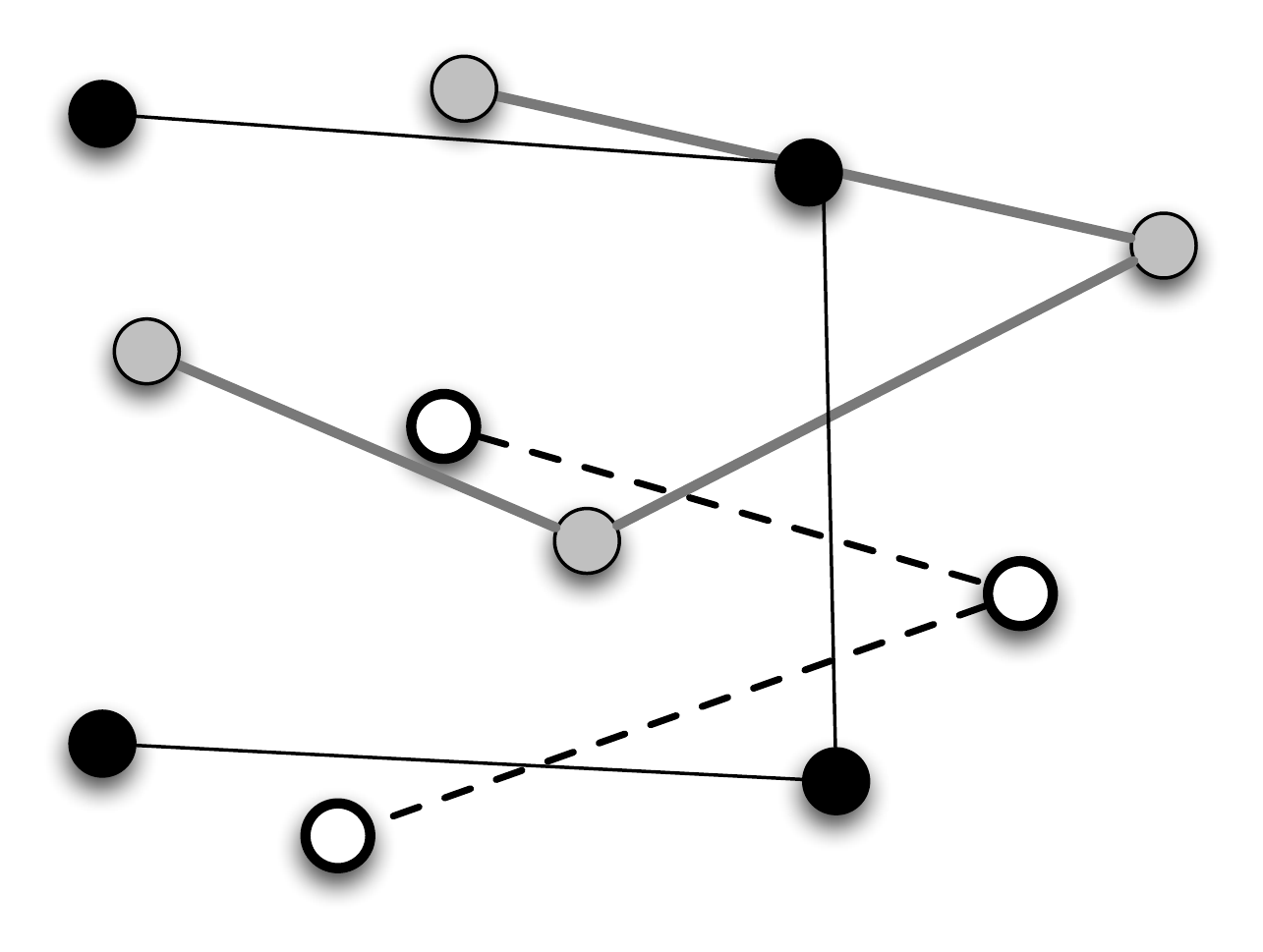}
\caption{TSP-Path for $V_0, V_1, V_2$.}
\label{fig:1c}
\end{subfigure}%
\hfill
\begin{subfigure}[b]{.21\linewidth}
\centering
\includegraphics[width=\linewidth]{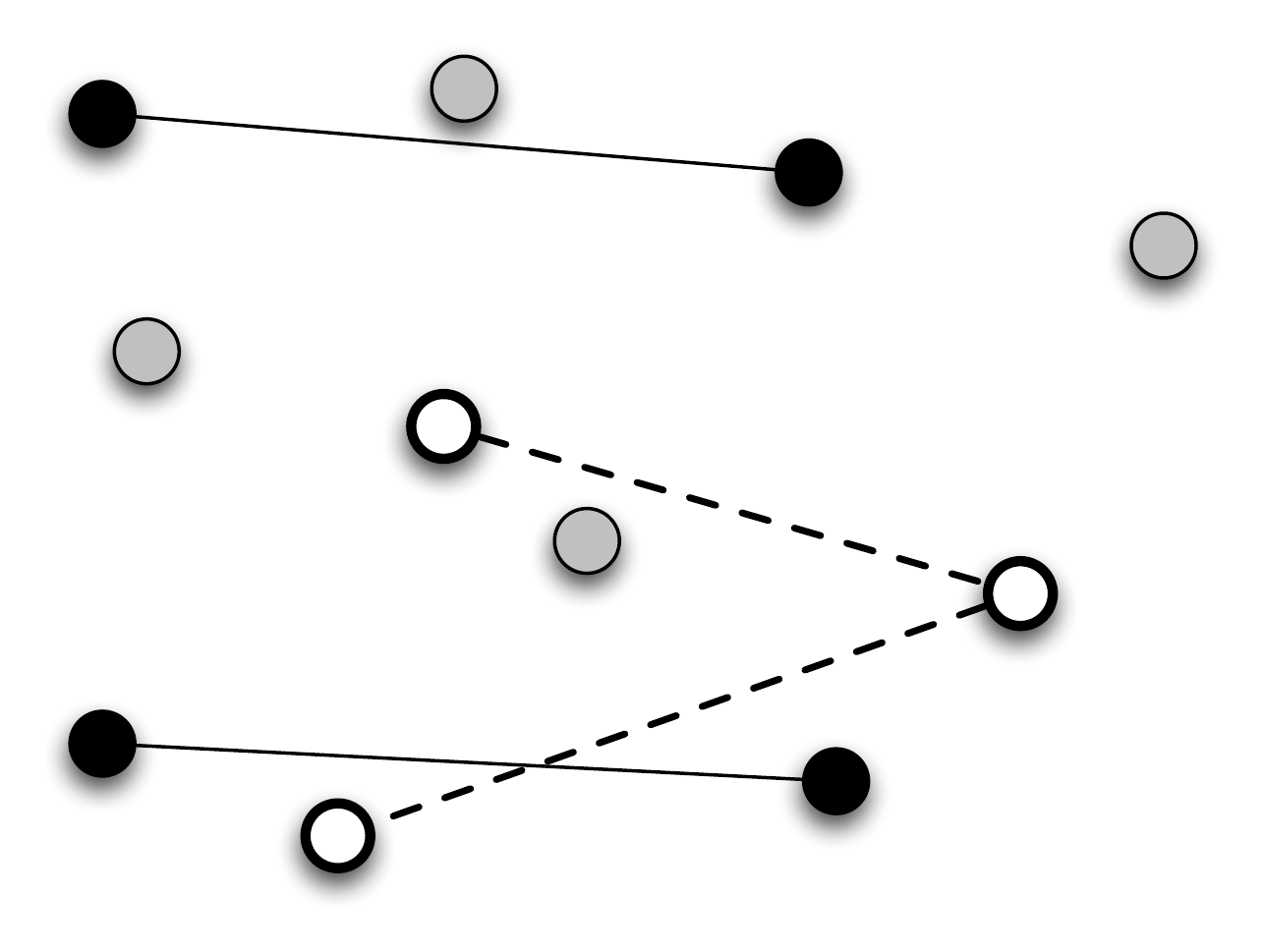}
\caption{Partitions for $V_0, V_1,V_2$.}
\label{fig:1d}
\end{subfigure}%

\vskip2em
\begin{subfigure}[b]{.21\linewidth}
\centering
\includegraphics[width=\linewidth]{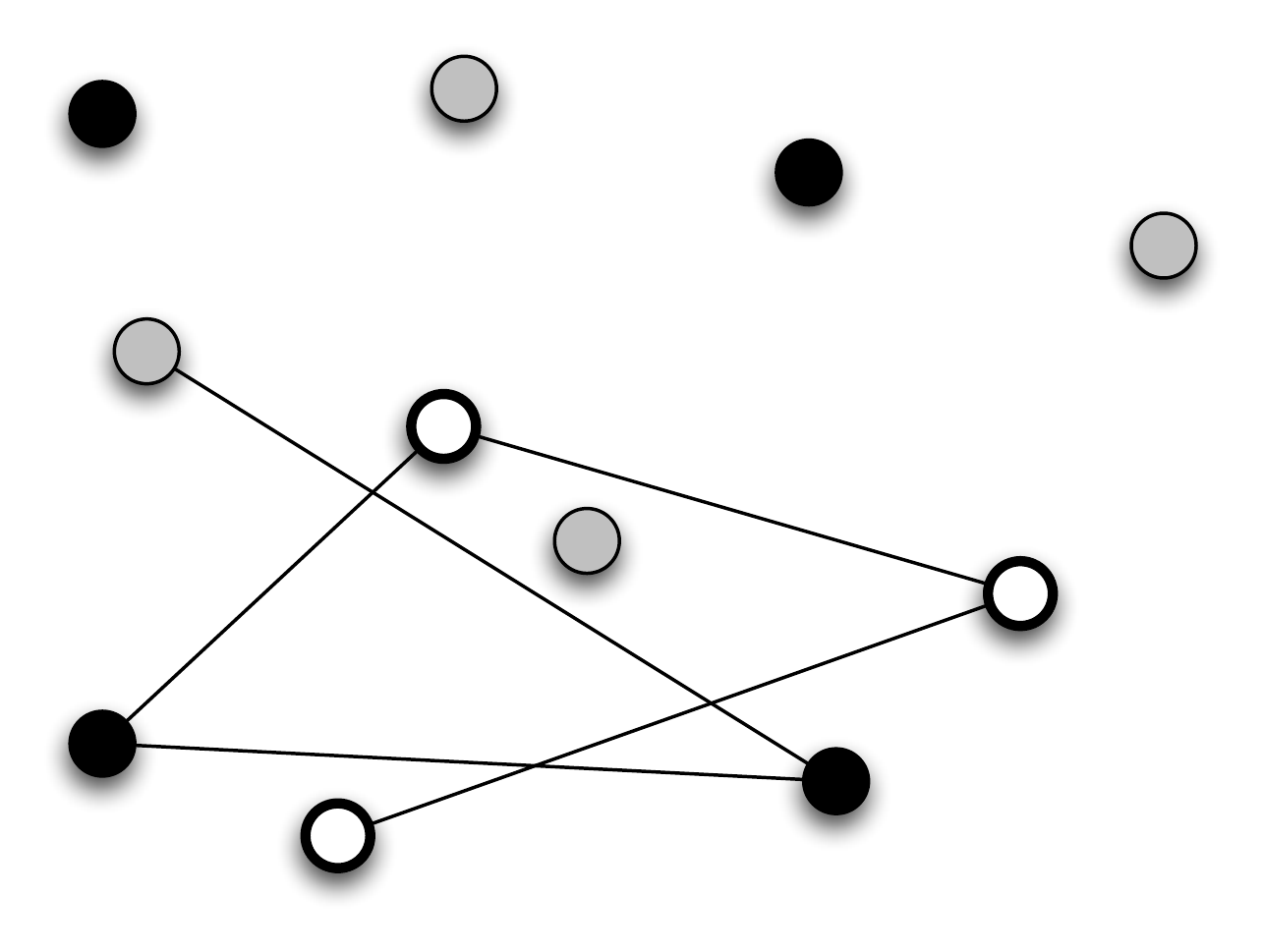}
\caption{The walk $S_1$.}
\label{fig:1e}
\end{subfigure}%
\hfill
\begin{subfigure}[b]{.21\linewidth}
\centering
\includegraphics[width=\linewidth]{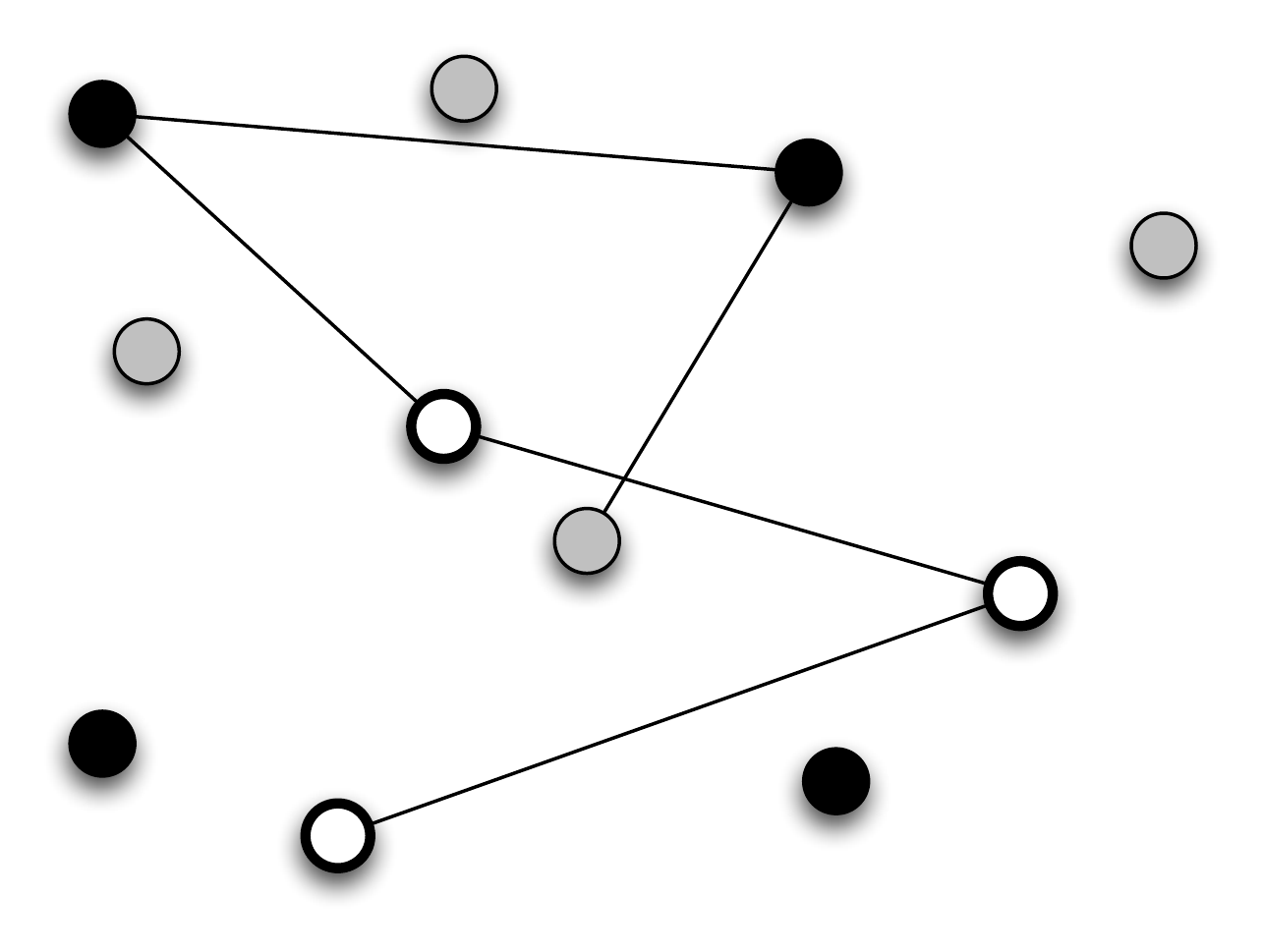}
\caption{The walk $S_2$.}
\label{fig:1f}
\end{subfigure}%
\hfill
\begin{subfigure}[b]{.21\linewidth}
\centering
\includegraphics[width=\linewidth]{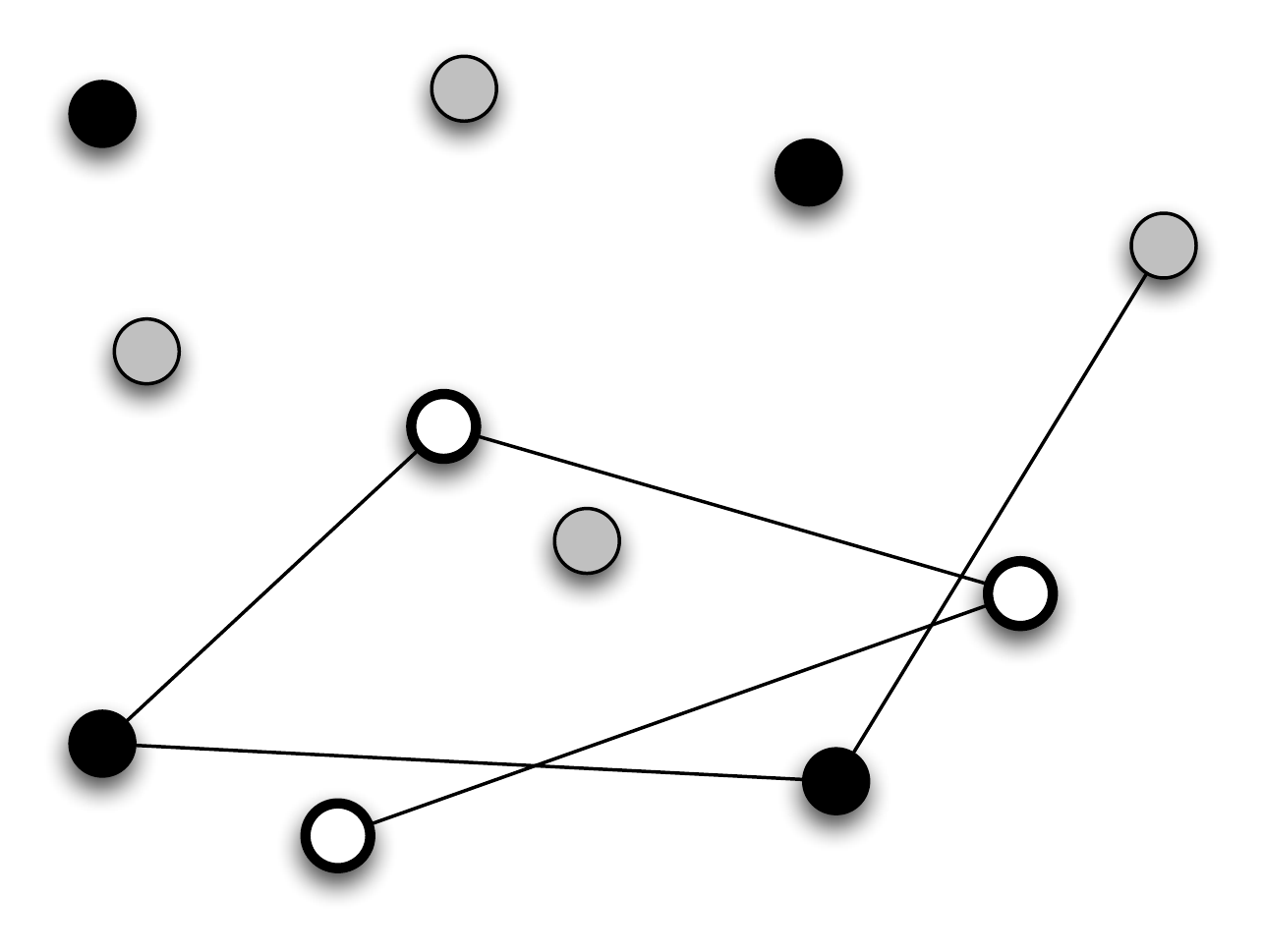}
\caption{The walk $S_3$.}
\label{fig:1g}
\end{subfigure}%
\hfill
\begin{subfigure}[b]{.21\linewidth}
\centering
\includegraphics[width=\linewidth]{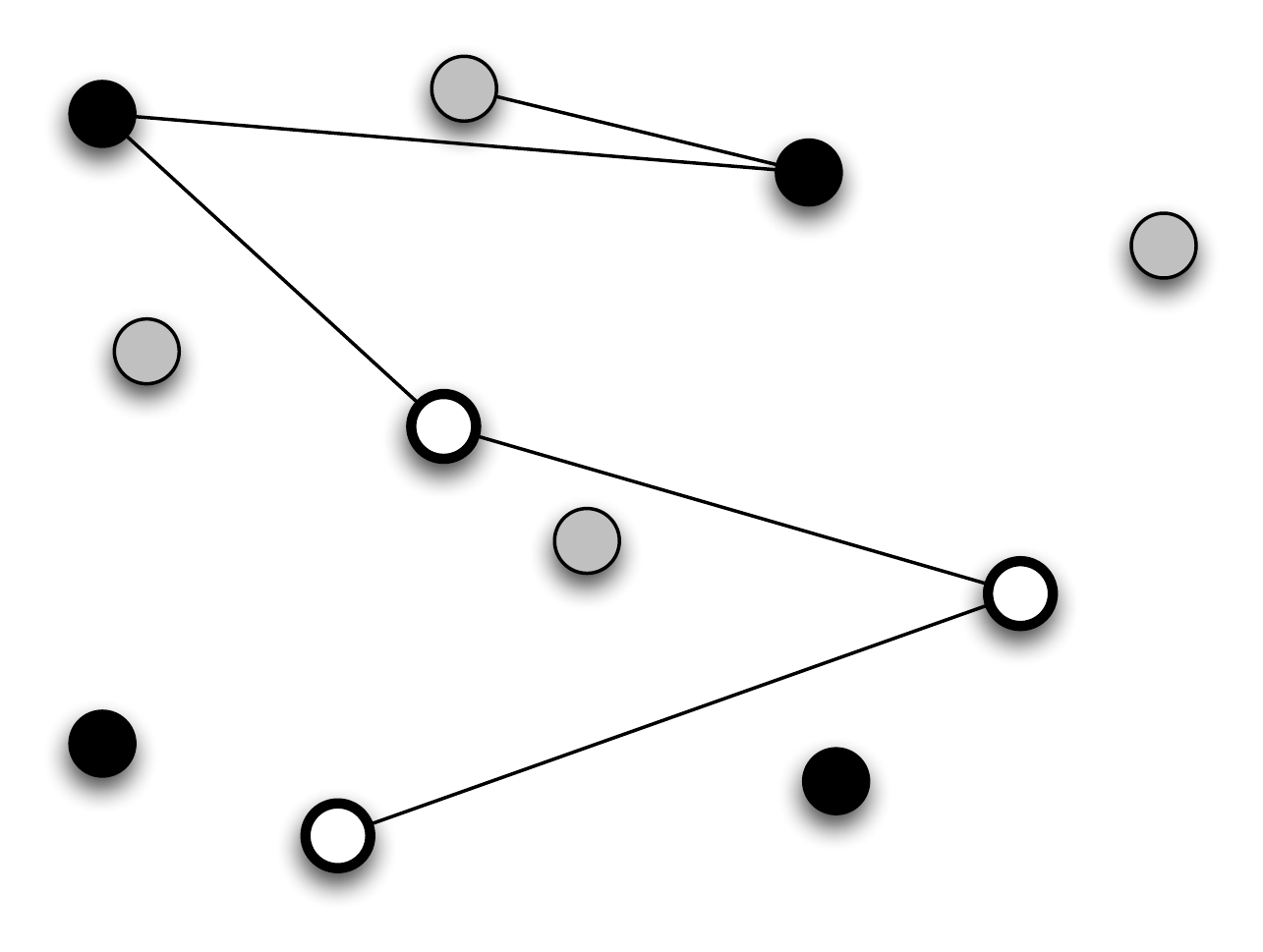}
\caption{The walk $S_4$.}
\label{fig:1h}
\end{subfigure}%

\vskip2em
\begin{subfigure}[b]{.21\linewidth}
\centering
\includegraphics[width=\linewidth]{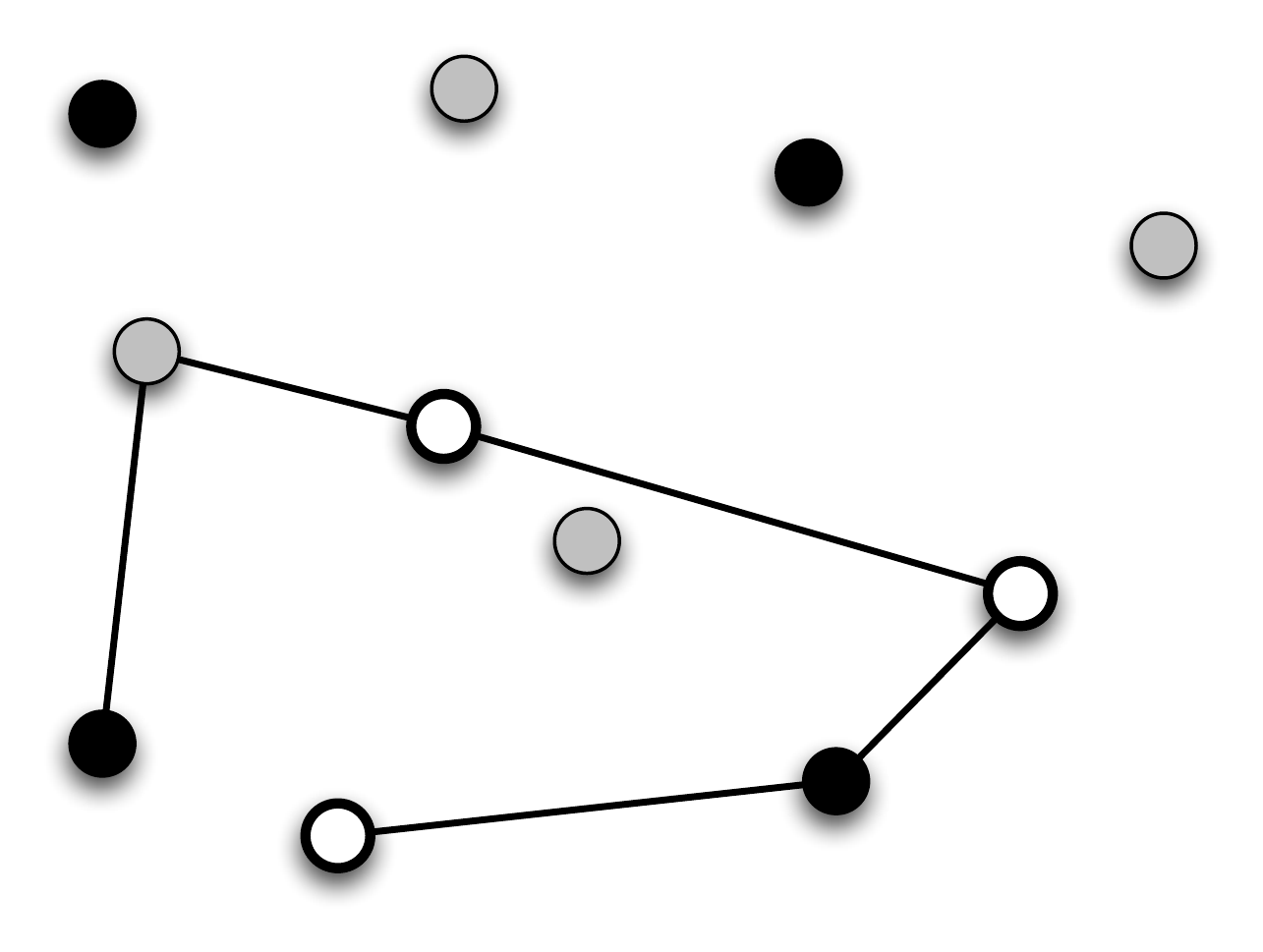}
\caption{TSP-Path of $S_1$.}
\label{fig:1i}
\end{subfigure}%
\hfill
\begin{subfigure}[b]{.21\linewidth}
\centering
\includegraphics[width=\linewidth]{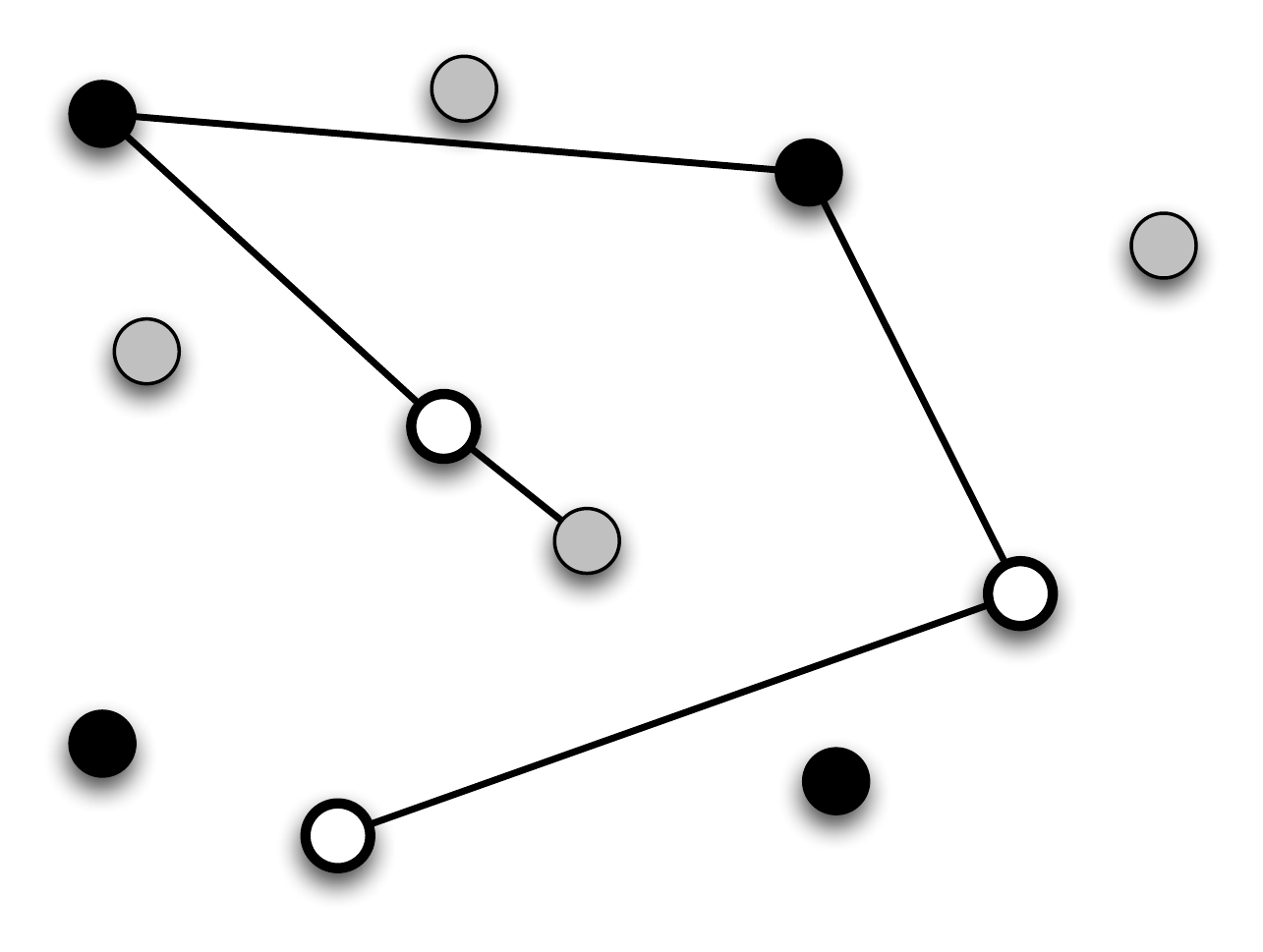}
\caption{TSP-Path of $S_2$.}
\label{fig:1j}
\end{subfigure}%
\hfill
\begin{subfigure}[b]{.21\linewidth}
\centering
\includegraphics[width=\linewidth]{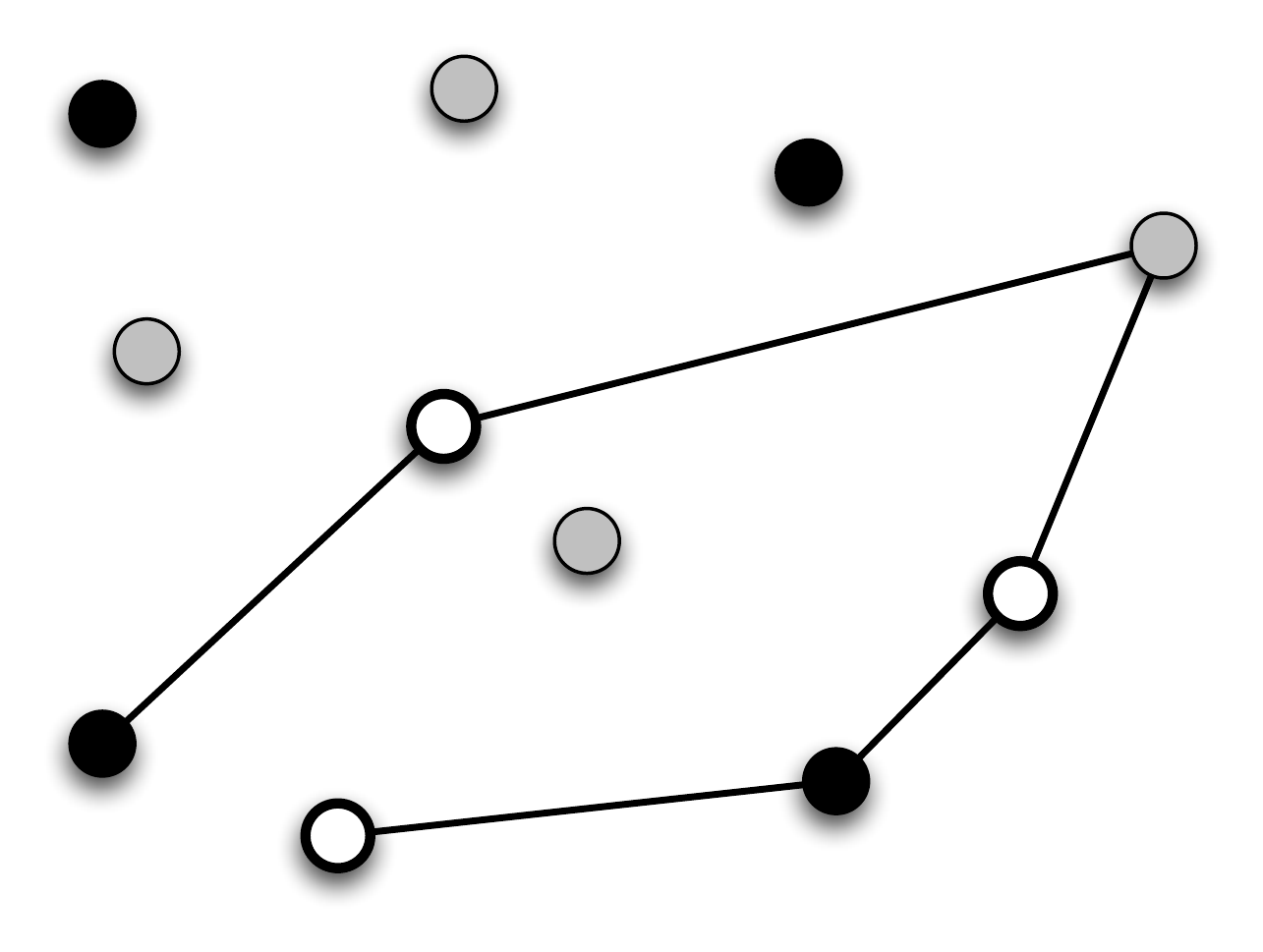}
\caption{TSP-Path of $S_3$.}
\label{fig:1k}
\end{subfigure}%
\hfill
\begin{subfigure}[b]{.21\linewidth}
\centering
\includegraphics[width=\linewidth]{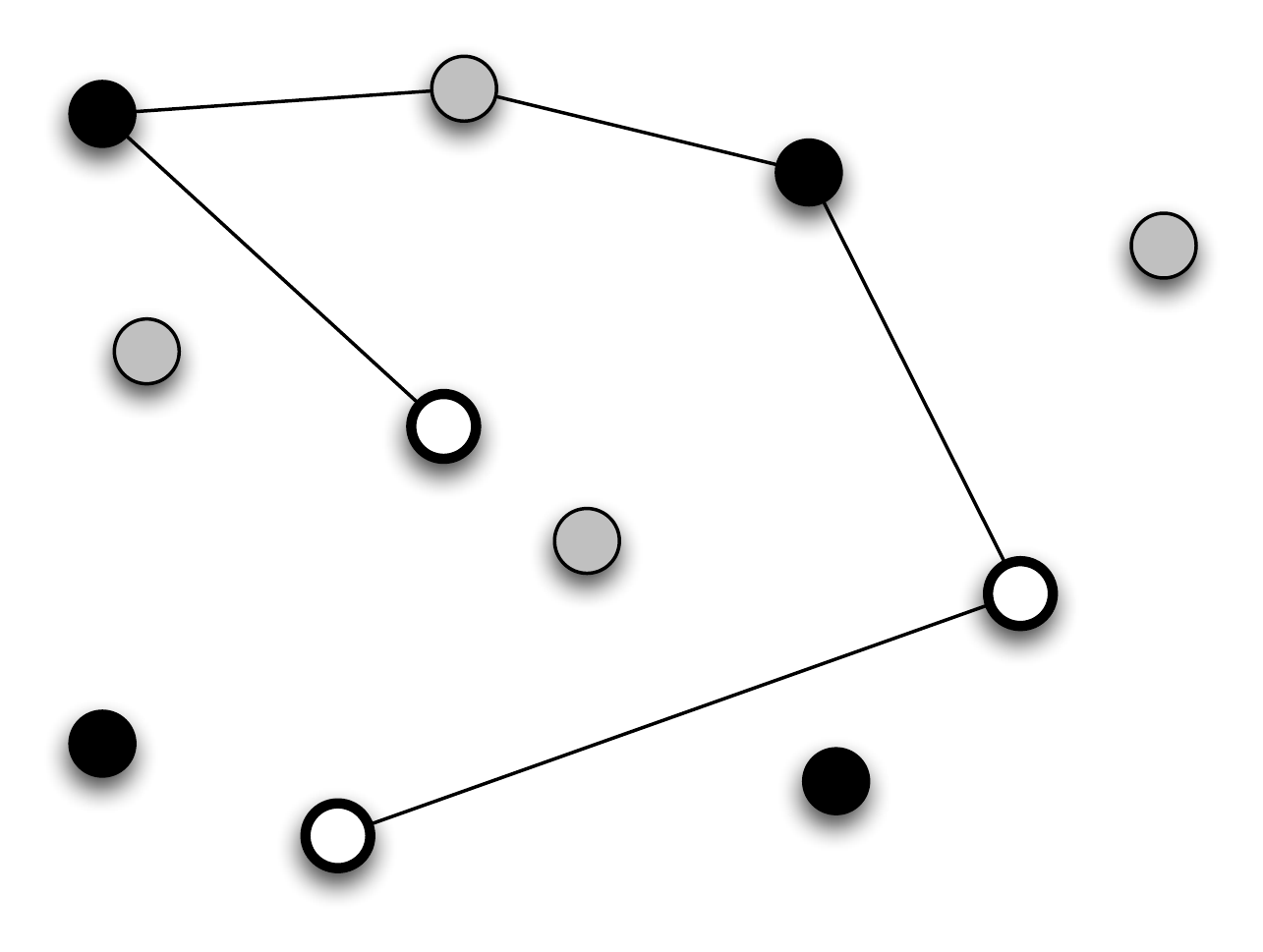}
\caption{TSP-Path of $S_4$.}
\label{fig:1l}
\end{subfigure}%
\caption{An illustration of the \textsc{BrutePartitionAlg} in Algorithm~\ref{alg:const_logarithmic_partition_epsilon}.  Figure (b) shows a relaxed graph consisting of vertices of weight $1$, $1/2$, and $1/4$.  The edge lengths are given by the Euclidean distance between vertices.  Figures (c-d) show the partitioning of the vertices that occurs in line 5.  Figures (e-h) show the four walks obtained in line 7.  Figures (i-l) shows how the walks can be heuristically improved by recomputing the TSP-Path through each of them, as discussed in Remark~\ref{rem:alg_imp}. }
\label{fig:Brute_alg}
\end{figure}

The following remark discusses a heuristic improvement that should be used in practice.

\begin{remark}[Implementation of Algorithm~\ref{alg:const_logarithmic_partition_epsilon}]
\label{rem:alg_imp}
In practice, we can recompute TSP-Path through each walk $S_1,\ldots,S_t$.  We simply require that each walk $\text{TSP-Path}(S_i)$ starts at the same vertex. In recomputing the TSP-Paths, the bounds remain unchanged.  However, in practice the performance is improved.  An example of the improvement is shown in Figure~\ref{fig:Brute_alg}.  Thus, the final infinite walk is obtained by repeatedly performing the walks in Figure~\ref{fig:1i} through to Figure~\ref{fig:1l}.
\end{remark}

\subsection{An $O(\log n)$-Approximation Algorithm}

In many applications, the value of $\rho_G$ is independent of $n$.  For example, in a robotic monitoring scenario, there may be only a finite number of importance levels that can be assigned to a point of interest.  In this case we have a constant factor approximation algorithm. However, the ratio between the largest and the smallest weights, i.e.,  $\rho_G$, does not directly depend on the size of the input graph. For even a small graph, $\rho_G$ can be very large. Therefore, in such cases we need an algorithm with an approximation guarantee that is bounded by a function of the size of the graph.

\begin{algorithm}
\caption{$\textsc{SmartPartitionAlg}(G)$}
\begin{algorithmic}[1]
    \STATE Let $V_i$ be the set of vertices of weight  {$\frac{1}{2^{i}} \le \phi(u) < \frac{1}{2^{i-1}}$} for $0 \leq i \leq \lceil\log_2 \rho_G\rceil$
    \STATE $U \leftarrow \cup_{i> \lfloor \log_2 n \rfloor+1} V_i$
    \STATE $S  {\leftarrow \textsc{BrutePartitionAlg}(G[V\backslash U])}$
    \STATE $k \leftarrow 1$
    \FORALL{$u\in V_i$ where $i > \lfloor \log_2 n \rfloor+1$}
        \STATE Insert $u$ at the end of $S_{2k}$
        \STATE Increment $k$
    \ENDFOR
    \RETURN $S$
\end{algorithmic}
\label{alg:const_logarithmic_partition_n}
\end{algorithm}

Our second approximation algorithm is shown in Algorithm~\ref{alg:const_logarithmic_partition_n}.  This algorithm is guaranteed to find a solution with cost within a logarithmic factor of the optimal cost. The main idea in Algorithm~\ref{alg:const_logarithmic_partition_n} is to construct a graph $G''$ from the input graph $G$ by relaxing weights and removing vertices with very low weights from $G$.  The result is that $\rho_{G''}$ is a function of $n$. Performing Algorithm~\ref{alg:const_logarithmic_partition_epsilon} on $G''$ results in a binary walk $S$ such that $\Delta(S)$ is an $O(\log n)$-factor approximation for the min-max latency walk problem in $G''$. We then insert the vertices in $V(G)\setminus V(G'')$ into $S$ in such a way that we maintain the $O(\log n)$-factor approximation on the input graph $G$.

\begin{theorem}{Given a graph $G$}, Algorithm \ref{alg:const_logarithmic_partition_n} constructs a walk $S$ of size in $O(n^2)$ such that $\Delta(S)$ is an $O(\log n)$-approximation for the min-max latency walk problem in $G$.
\end{theorem}
\begin{proof}
  The idea of Algorithm~\ref{alg:const_logarithmic_partition_n} is to remove the vertices of small weight so that we can use Algorithm~\ref{alg:const_logarithmic_partition_epsilon} as a subroutine. Let $G'$ be the result of relaxing the weights of $G$ and $U$ be the set of vertices of $G'$ with weight less than $1/2^{\lfloor \log_2 n \rfloor +1}$, as in line 2 of Algorithm~\ref{alg:const_logarithmic_partition_n}. Let $G''$ be the result of removing vertices in $U$ from $G'$. As in Line 3 of Algorithm~\ref{alg:const_logarithmic_partition_n}, walk $S=[S_1,S_2,\ldots,S_t]$ is the result of running Algorithm \ref{alg:const_logarithmic_partition_epsilon} on $G''$ with $n<\rho_{G''}=2^{\lfloor \log_2 n \rfloor +1}\le 2n$.  Recall that in Algorithm~\ref{alg:const_logarithmic_partition_epsilon}, for $1\le k\le t$ each $S_k$ starts with the same vertex $v$ with $\phi(v)=1$. Moreover, by the proof of Theorem \ref{tho:method1}, since $\rho_{G''}\le2n$ each $S_k$ has length at most $2\log_2 \rho_{G''}\opt_{G''}=(2\log_2 n + 2)\opt_{G''}$. Also, since $\opt_{G''}\le \opt_{G'}$, we have that $(2\log_2 n + 2)\opt_{G'}$ is an upper-bound for lengths of $S_k$'s.

In line 6 of Algorithm~\ref{alg:const_logarithmic_partition_n}, the $k$-th vertex of $U$ is inserted at the end of walk $S_{2k}$. Note that since $|U|<n$ and $2n < t=2\rho_{G''}$ this is possible.
Since each walk $S_k$ begins in vertex $v$ with $\phi(v)=1$, by Lemma \ref{lem:maxdishalfc} and Corollary \ref{cor:maxdisc}, each detour to a vertex in $U$ has length bounded by $\frac{3}{2}\opt_{G'}$. Consequently, after inserting vertices in $U$ into $S$, each $S_k$ has length at most $(2\log_2 n + \frac{7}{2})\opt_{G'}$ in $G'$. Hence, by Lemma \ref{lem:perfectdecomposition}, the cost of $\Delta(S)$ in $G'$ is at most $(4\log_2 n + 8)\opt_{G'}$. Furthermore, the cost of $\Delta(S)$ in $G$ is bounded by $(8\log_2 n + 16)\opt_{G}$ using Lemma \ref{lem:relaxing}. As a result, $\Delta(S)$ is an $O(\log n)$-approximation for the min-max latency walk problem in $G$.

Moreover, since $S$ is a binary walk, its size is bounded by $2n\rho_{G''}\le4n^2$. Since the number of vertices added to $S$ during its modification is in $O(n)$, the size of the constructed walk $S$ is in $O(n^2)$.
\end{proof}

As mentioned in Remark~\ref{rem:alg_imp}, we can improve the performance of Algorithm~\ref{alg:const_logarithmic_partition_n} by computing TSP-Paths through each of the modified walks in $S$, making sure that all paths start at the same vertex.  The following result shows that Algorithm~\ref{alg:const_logarithmic_partition_n} always achieves a better approximation factor than Algorithm~\ref{alg:const_logarithmic_partition_epsilon}.

\begin{corollary}
\label{cor:3betterthan2}
The approximation ratio of Algorithm~\ref{alg:const_logarithmic_partition_n} is $O\big(\log \min\{n,\rho_G\} \big)$, which is always less than or equal to that of Algorithm~\ref{alg:const_logarithmic_partition_epsilon}.
\end{corollary}
\begin{proof}
In Algorithm~\ref{alg:const_logarithmic_partition_n}, if for the input graph $G$ we have $\rho_G\le 2n$, then the set $U$ is empty and hence the approximation guarantees of both Algorithms~\ref{alg:const_logarithmic_partition_epsilon} and \ref{alg:const_logarithmic_partition_n} are $O(\log \rho_G)$. On the other hand, if $\rho_G> 2n$ then the approximation factor of Algorithm~\ref{alg:const_logarithmic_partition_n}, i.e., $O(\log n)$, is smaller than the approximation factor of Algorithm~\ref{alg:const_logarithmic_partition_epsilon}, which is $O(\log \rho_G)$.
\end{proof}

\section{Simulations}
\label{sec:simulations}

In this section we present two sets of simulations.  The first shows the scalability of our algorithms with respect to the size of the graph, and compares the performance to a simple TSP-based algorithm.  The second set shows a case study for patrolling for crime in the city of San Francisco, CA.

\subsection{Scalability of Approximation Algorithms}

\begin{figure}[t]
\centering
\includegraphics[width=0.8\linewidth]{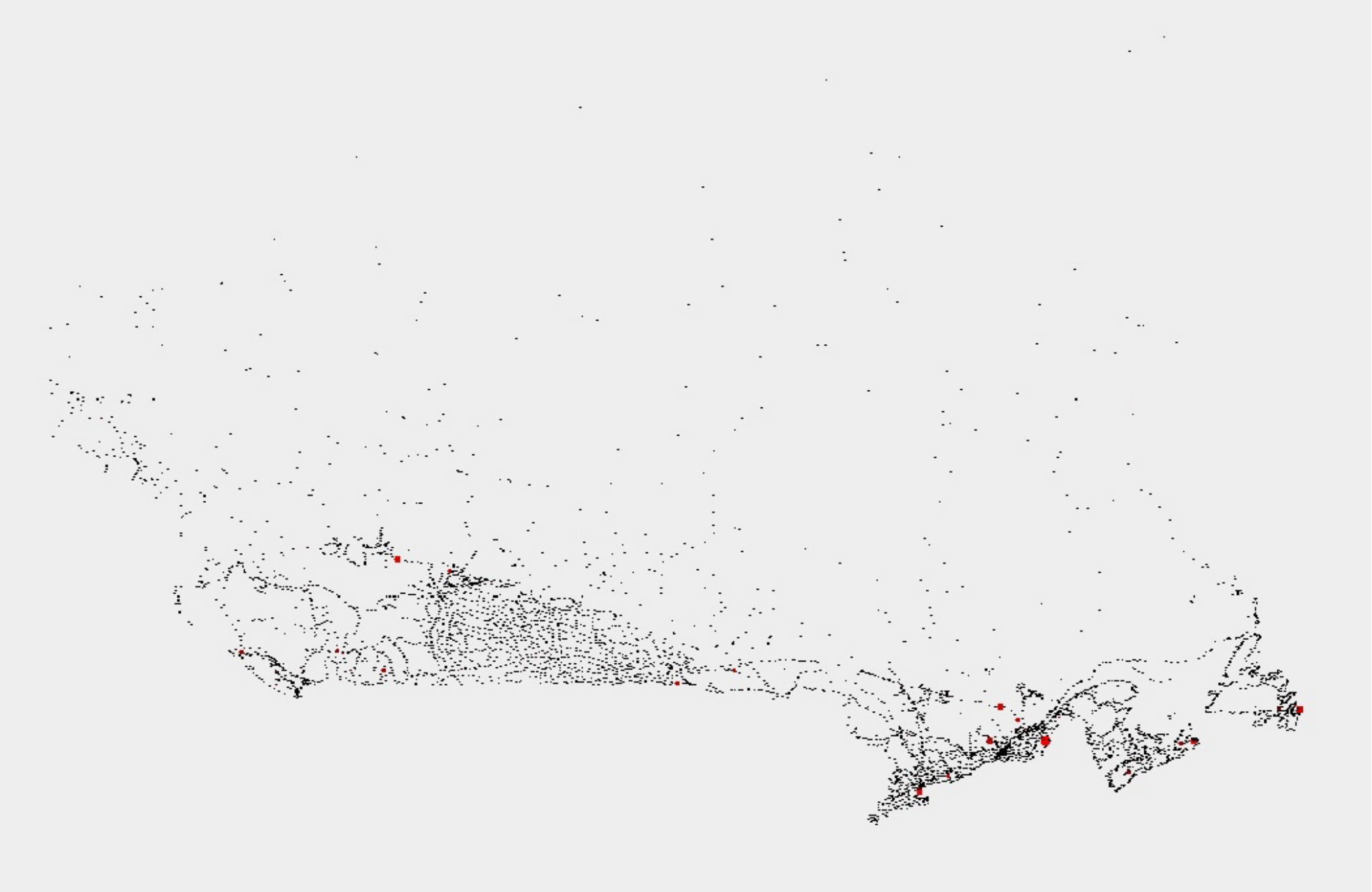}
%
%
\caption{The 4663 vertex graph used for all tests corresponding to all cities in Canada~\cite{CWJ:09}.}
\label{fig:Canada}       
\end{figure}

In this section, we study the scalability of Algorithm~\ref{alg:const_logarithmic_partition_n}.  Recall that by Corollary~\ref{cor:3betterthan2}, Algorithm~\ref{alg:const_logarithmic_partition_n} always performs better than Algorithm~\ref{alg:const_logarithmic_partition_epsilon}, both in runtime and approximation factor.

For the simulations, we use test data that are standard benchmarks for testing the performance of heuristic algorithms for the TSP. The data sets used here are taken from~\cite{CWJ:09}. Each data set represents a set of locations in a country. We construct a graph by placing a vertex for each location and letting the length of the edge connecting each pair of vertices be the Euclidean distance of the corresponding points.

What remains is to assign weights to each vertex in the test cases.  In many persistent monitoring applications, the likelihood of a vertex with very high weight is low. In other words, majority of vertices have low priority, while few vertices need to be visited more frequently. To simulate this behavior, we use a distribution that has the following exponential property:
\begin{equation}
\label{equ:prob}
\mathrm{P}\left[\left(\frac{1}{2}\right)^{k+1} < \phi(v) \leq \left(\frac{1}{2}\right)^k\right] = \frac{1}{B},
\end{equation}
for $k<B$, where $B$ is a fixed integer. To create such a distribution, we assign to each vertex $v$ a weight {${(1/2)}^B\leq \phi(v)\leq 1$} with probability  {$f\big(\phi(v)\big) = {\big(\phi(v) B\ln2\big)}^{-1}$}.

Here, we compare our algorithms to the simple algorithm of finding a TSP tour through all vertices in $G$. For finding an approximate solution to the TSP we use an implementation of the Lin-Kernighan algorithm~\cite{LS-KB:73} available at~\cite{CWJ:09}. Recall that from Lemma~\ref{lem:TSPbad}, the walk obtained from the TSP can have a cost that is $n$ times larger than the optimal cost.  However, when all vertex weights are equal, the TSP yields the optimal walk.  In simulation, when the vertex weights are uniformly distributed, the TSP appears to provide a fairly good approximation for the min-max latency walk problem.  One of the reasons for this is that when the vertex weights are distributed uniformly, we expect that half of the vertices will have weights in $[0.5,1]$.  Thus $O(n)$ of the vertices must be visited very frequently, and not much can be gained by visiting vertices at different frequencies. 


\paragraph{Performance with Respect to Vertex Weight Distribution:}

\begin{figure}
\centering
\includegraphics[width=0.7 \linewidth]{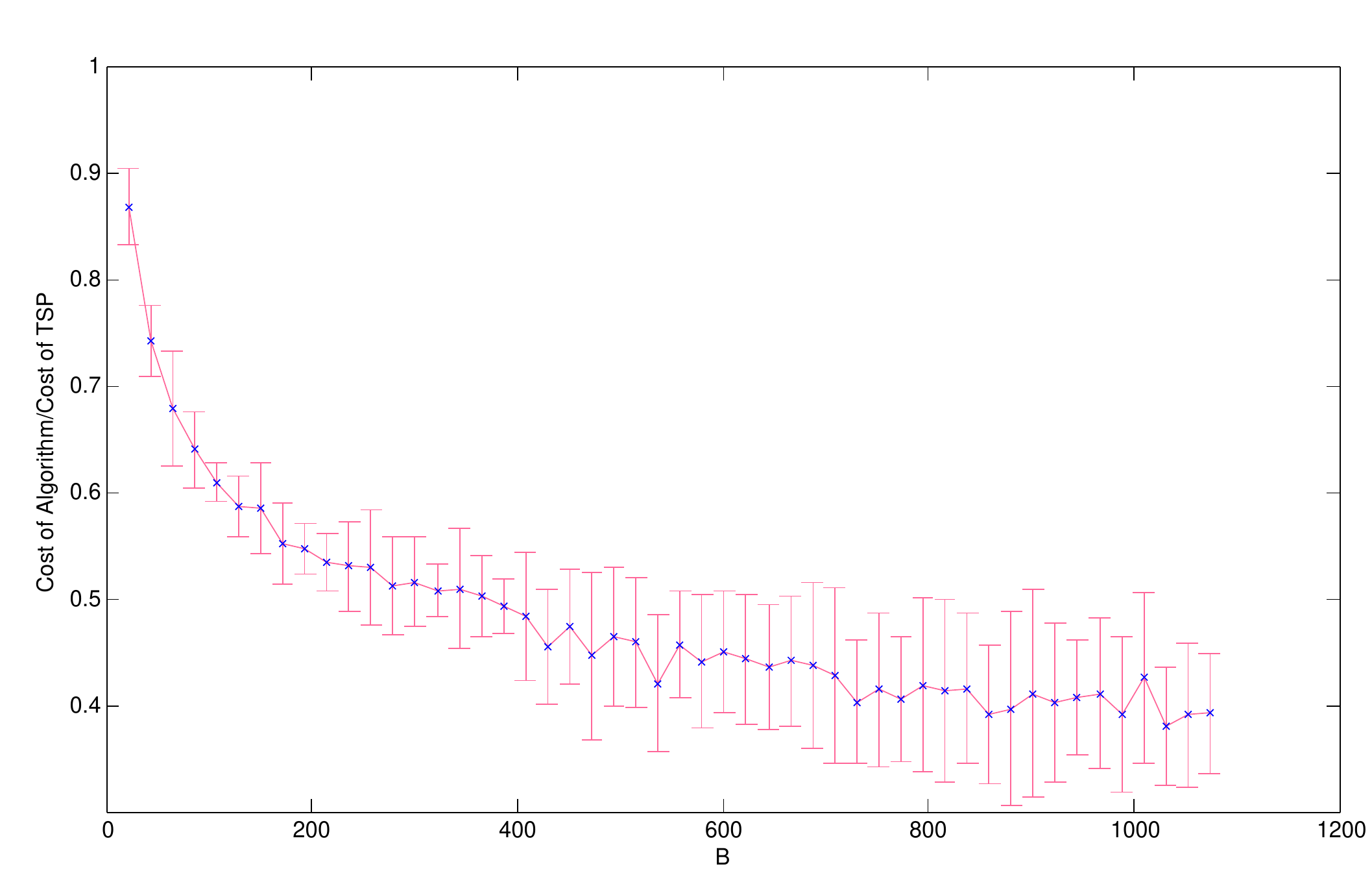}\\
\caption{The ratio of the cost of the walk produced by Algorithm \ref{alg:const_logarithmic_partition_n} to the cost of the TSP tour as a function of $B$ in~\eqref{equ:prob}.}
\label{fig:AlgOverB}
\end{figure}

An important aspect of an environment is the ratio of weight of its elements, therefore it is natural to test our algorithm with respect to $\rho_G$. Note that by Equation \ref{equ:prob}, $\rho_G > (1/2)^B$. Therefore, we consider different values of $B$ to assess the performance of the algorithm for different ranges of weights on the same graph (see Figure~\ref{fig:Canada}). It is easy to see that 
if $B < \log_2 n$, then Algorithms~\ref{alg:const_logarithmic_partition_n} and~\ref{alg:const_logarithmic_partition_epsilon} produce the same output. Figure~\ref{fig:AlgOverB} depicts the behavior of Algorithm~\ref{alg:const_logarithmic_partition_n} on a graph induced by 4663 cities in Canada, shown in Figure~\ref{fig:Canada}, with different values for $B$. It can be seen that as $B$ increases, our algorithm outperforms the TSP by a greater factor.

\paragraph{Performance with Respect to Input Graph Size:}

\begin{figure}
\centering
\includegraphics[width=0.7\linewidth]{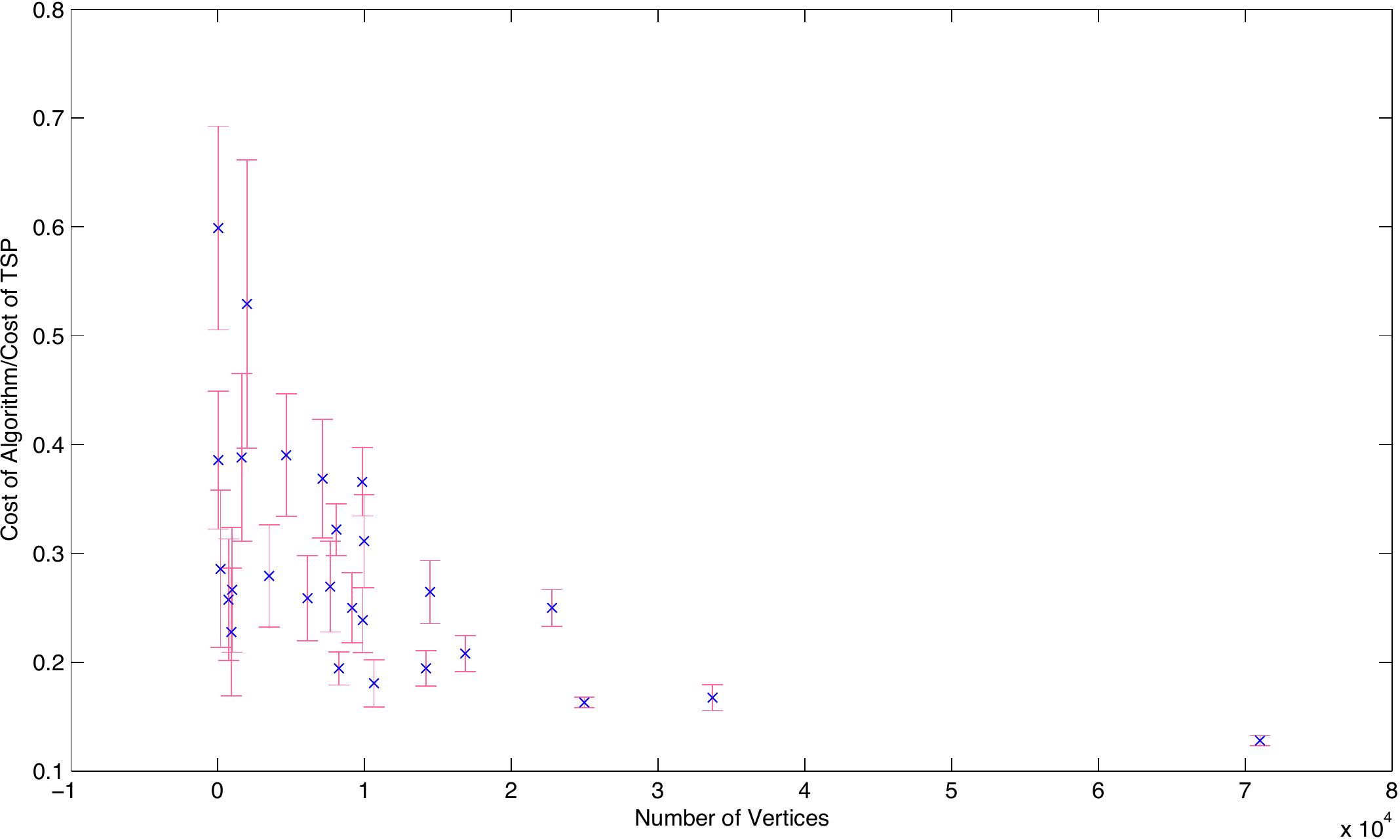}
\caption{The ratio of the cost of Algorithm~\ref{alg:const_logarithmic_partition_n} to the cost of the TSP tour on the $27$ test graphs in~\cite{CWJ:09}.}
\label{fig:AlgOverN}
\end{figure}

We use graphs of different sizes to evaluate the performance and scalability of our algorithms. Again, the cost is compared to that of a simple TSP tour that visits each vertex in the graph exactly once before returning to its starting vertex. Figure~\ref{fig:AlgOverN} depicts the ratio of the cost of the walk constructed by Algorithm~\ref{alg:const_logarithmic_partition_n} to the cost of the TSP tour on $27$ different graphs each corresponding to a set of locations in a different country. Here the value of $B$ is fixed to 1000. It can be seen that the ratio of the cost of the walk produced by our algorithm to the cost of the TSP tour decreases as the size of graph increases. Hence, as the number of vertices in graphs increases, the performance of Algorithm \ref{alg:const_logarithmic_partition_n} relative to the TSP improves.

Also, the time complexity of the algorithm is $O\big(n^2+\beta(n)\big)$ where $\beta(n)$ is the running time of the algorithm used for finding TSP tours. For the test data corresponding to 71009 locations in China, our Java implementation of Algorithm \ref{alg:const_logarithmic_partition_n} constructs an approximate solution in 20 seconds using a laptop with a 2.50 GHz CPU and 3 GB RAM.

\subsection{A Case Study in Patrolling for Crime}

Here we study the problem of planning a route for a robot (or vehicle) to patrol a city. Assume we want to plan a route for that patrols a set of intersections in a city, and the goal is to minimize the maximum expected number of crimes that occur at any of the intersections between two consecutive visits. If the weight of each intersection is given by the average crime rate in that intersection, then since expectation is a linear function, this problem will exactly translate to the min-max latency walk problem.

\begin{table}
\footnotesize
\begin{center}
\begin{tabular}{|c|c|l|}
  \hline
  Index & \# of crimes in Aug 2012 & Approximate address \\ \thickhline
  A & 133 &Sutter St \& Stockton St, San Francisco, CA 94108, USA\\ \hline
  B & 90 & Pacific Ave \& Grant Ave, San Francisco, CA 94133, USA\\
  C & 89 & Post St \& Taylor St, San Francisco, CA 94142, USA\\
  D & 87 & Jackson St \& Front St, San Francisco, CA 94111, USA\\
  E & 83 & Vallejo St \& Powell St, San Francisco, CA 94133, USA\\
  F & 83 & Bay St \& Mason St, San Francisco, CA 94133, USA\\
  G & 74 & Bush St \& Montgomery St, San Francisco, CA 94104, USA\\ \hline
  H & 64 & Bush St \& Hyde St, San Francisco, CA 94109, USA\\
  I & 48 & Chestnut St \& Montgomery St, San Francisco, CA 94111, USA\\
  J & 43 &Washington St \& Leavenworth St, San Francisco, CA 94109, USA\\
  K & 38 &Jones St \& Beach St, San Francisco, CA 94133, USA\\
  L & 34 &Hyde St \& Francisco St, San Francisco, CA 94109, USA\\
  \hline
\end{tabular}
\caption{Twelve locations in the central district of San Francisco police department with the number of recorded criminal acts in the vicinity of each location.}
\label{tab:locationDetails}
\end{center}
\end{table}

\begin{figure}
\centering
\includegraphics[scale=.4]{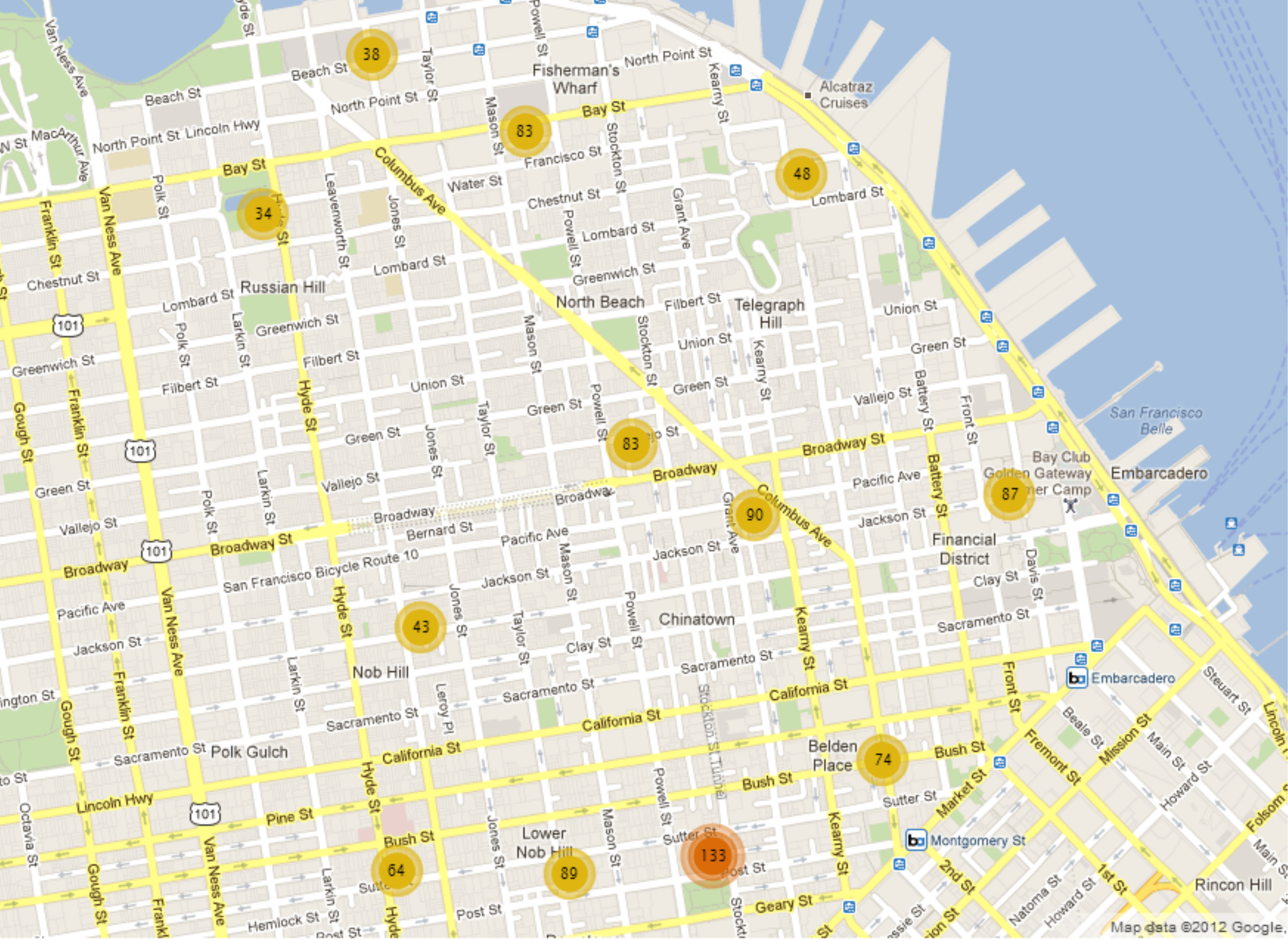}\\
\caption{The twelve intersections of Table \ref{tab:locationDetails} in North East of San Francisco.  Map obtained from San Francisco Police Department CrimeMAPS program.}
\label{fig:SanRaw}
\end{figure}

We look at twelve intersections in the central district of the San Francisco police department (see Table \ref{tab:locationDetails} and Figure \ref{fig:SanRaw}). The number of crimes that occurred in August of 2012\footnote{Crime data for San Francisco was obtained from the San Francisco Police Department CrimeMAPS website:  \url{http://www.sf-police.org/}.} in the vicinity of these intersections is used as the weight function $\phi$. The weight of each intersection approximates the average rate of crime happening in vicinity of that intersection.
Table \ref{tab:pairwiseDistance} shows the pairwise travel times (for a road vehicle) of these intersections in seconds. Note that these travel times are not, in general, symmetric as some streets are one-way. For $l(uv)$, we use the average between the travel time from $u$ to $v$ and the travel time from $v$ to $u$.

\begin{table}
\footnotesize
\begin{center}
\begin{tabular}{|l "l |l |l |l |l |l |l |l |l |l |l |l |}
\hline
     & A    & B    & C    & D     & E    & F      & G    & H    & I    & J   & K   & L   \\ \thickhline
A	&0	&	141&	121&	293&	209&	329&	134&	250&	406&	199&	358&	344\\ \hline	
B	&141	&	0  &	271&	200&	105&	226&	201&	299&	297&	169&	254&	274\\ \hline	
C	&127	&	291&	0  &	368&	311&	433&	153&	198&	491&	219&	461&	362\\ \hline	
D	&304	&	207&	417&	0  &	253&	309&	226&	387&	249&	358&	337&	384\\ \hline	
E	&210	&	147&	340&	244&	0  &	180&	244&	268&	342&	164&	209&	230\\ \hline	
F	&330	&	216&	460&	244&	175&	0  &	313&	370&	126&	311&	61 &	163\\ \hline	
G	&90	&	246&	162&	244&	310&	369&	0  &	271&	400&	292&	397&	427\\ \hline	
H	&147	&	293&	105&	370&	338&	412&	154&	0  &	492&	153&	406&	287\\ \hline	
I	&426	&	324&	539&	203&	343&	226&	348&	509&	0  &	448&	299&	389\\ \hline	
J	&201	&	170&	231&	322&	164&	290&	279&	159&	415&	0  &	283&	164\\ \hline	
K	&354	&	240&	474&	337&	199&	105&	337&	332&	226&	273&	0  &	125\\ \hline	
L	&334	&	220&	354&	316&	179&	121&	317&	212&	246&	153&	114&	0\\	\hline
\end{tabular}
\caption{The pairwise by-car travel times in seconds corresponding to the locations in Table \ref{tab:locationDetails}, queried from Google maps.}
\label{tab:pairwiseDistance}
\end{center}
\end{table}

Let us step through Algorithm \ref{alg:const_logarithmic_partition_n}. Since $\log \rho_G < \lfloor \log n \rfloor + 1$, Algorithms \ref{alg:const_logarithmic_partition_n} simply returns the output of Algorithm \ref{alg:const_logarithmic_partition_epsilon}. We have $V_0 = \{\mathrm{A}\}$, $V_1 = \{\mathrm{B\;,C,\;D,\;E,\;F,\;G}\}$, and $V_2 = \{\mathrm{H,\;I,\;J,\;K,\;L}\}$. Then we find the TSP-Paths, which are given by $W_0=(\mathrm{A})$, $W_1=(\mathrm{C,\;G,\;D,\;F,\;E,\;B})$, and $W_2=(\mathrm{I,\; K,\; L,\; J,\; H})$. We then partition these into $W_{0,0}=(\mathrm{A})$, $W_{1,0}=(\mathrm{C,\;G,\;D})$, $W_{1,1}=(\mathrm{F,\;E,\;B})$, $W_{2,0}=(\mathrm{I})$, $W_{2,1}=(\mathrm{K,\; L})$, $W_{2,2}=(\mathrm{J})$, and $W_{2,3}=(\mathrm{H})$. Then, by concatenation of these walks and finding the TSP-Path of the results we get $S_1 = (\mathrm{A,\; C,\; G,\; D,\; I})$, $S_2 = (\mathrm{A,\; B,\; L,\; K,\; F,\; E})$, $S_3 = (\mathrm{A,\; C,\; J,\; D,\; G})$, and $S_4 = (\mathrm{A,\; B,\; E,\; F,\; H})$.  The walks are shown in Figure~\ref{fig:crime_tours}.  Note that the TSP-Path is shown instead of the TSP tour, and thus the final edge that returns to intersection A is omitted.  This is done so that the output is a binary walk, i.e., each walk $S_i$ contains at most one instance of each intersection.

\begin{figure}
\begin{subfigure}[b]{0.49\linewidth}
\centering
\includegraphics[width=\linewidth]{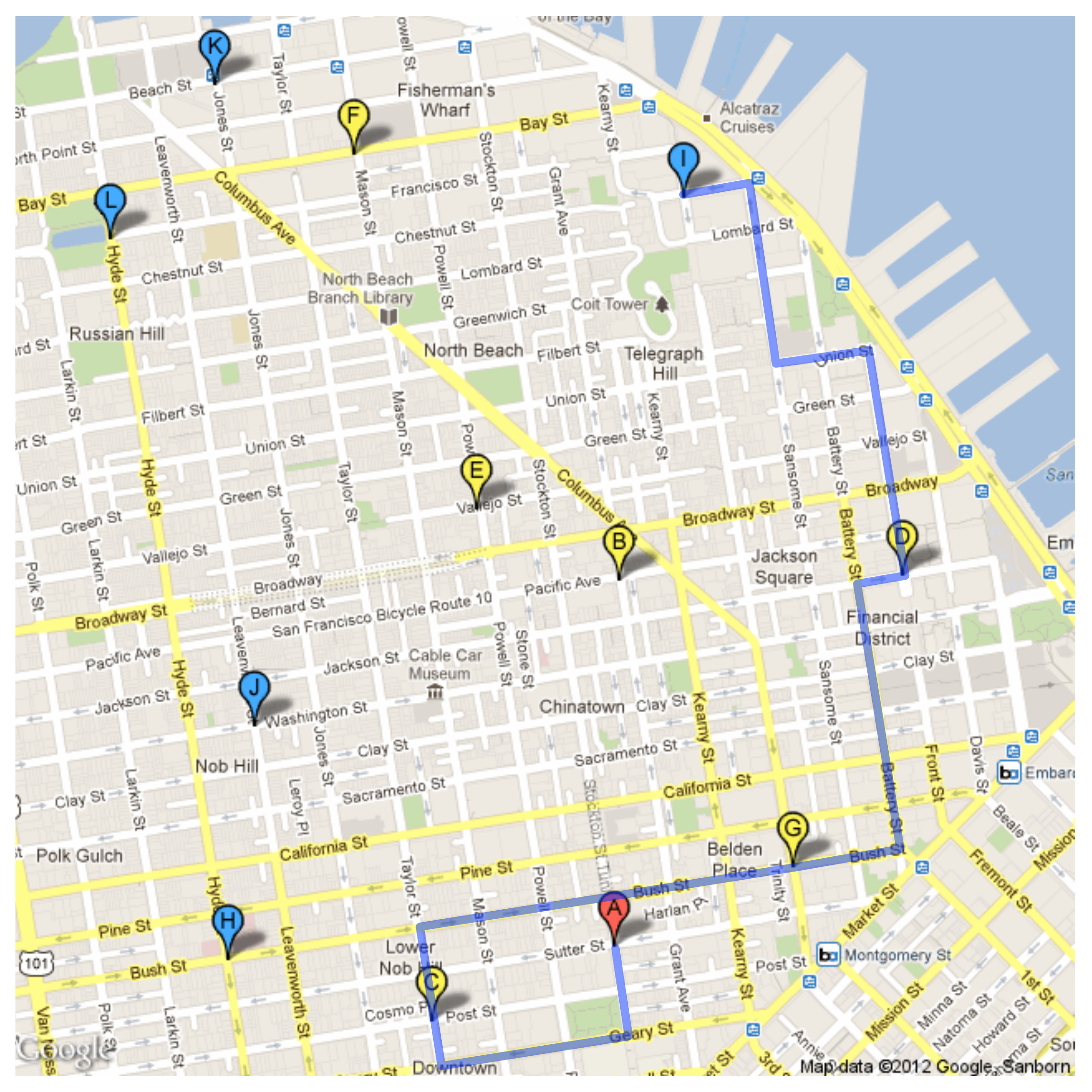}
\caption{The walk $S_1 = (\mathrm{A,\; C,\; G,\; D,\; I})$.}
\label{fig:TheS1}
\end{subfigure}
\hfill
\begin{subfigure}[b]{0.49\linewidth}
\centering
\includegraphics[width=\linewidth]{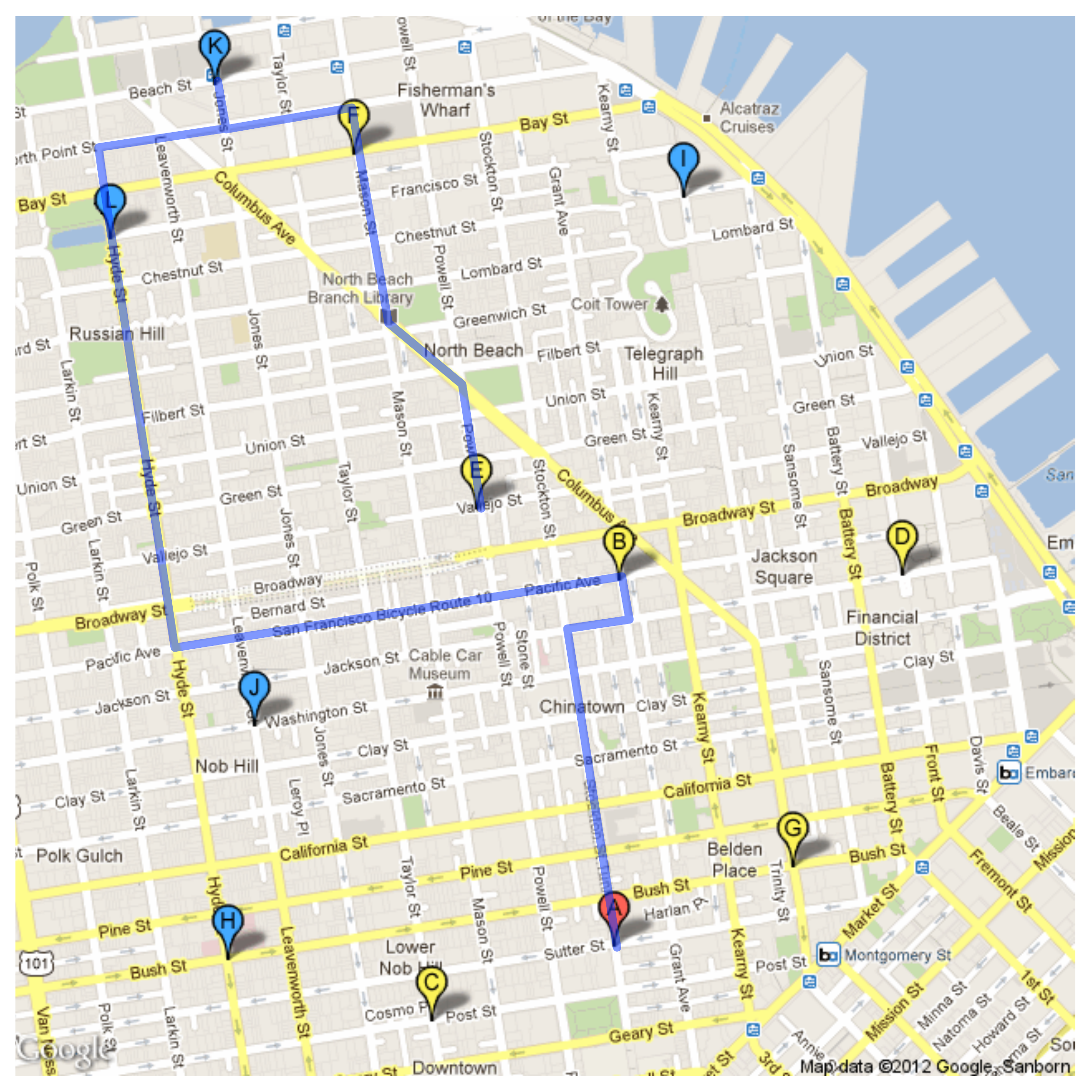}
\caption{The walk $S_2 = (\mathrm{A,\; B,\; L,\; K,\; F,\; E})$.}
\label{fig:TheS2}
\end{subfigure}
\vskip1em
\begin{subfigure}[b]{0.49\linewidth}
\centering
\includegraphics[width=\linewidth]{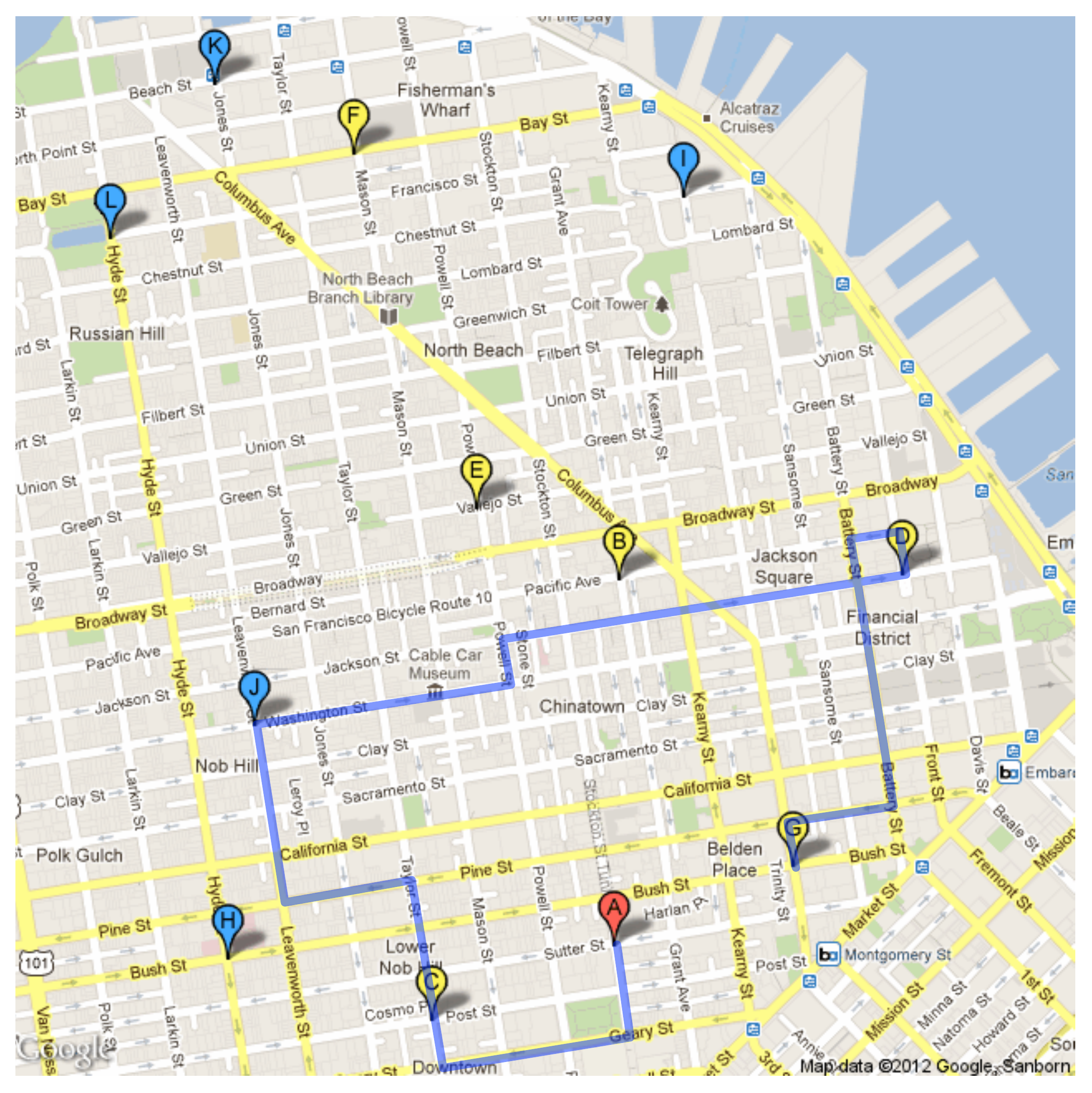}
\caption{The walk $S_3 = (\mathrm{A,\; C,\; J,\; D,\; G})$.}
\label{fig:TheS3}
\end{subfigure}
\hfill
\begin{subfigure}[b]{0.49\linewidth}
\centering
\includegraphics[width=\linewidth]{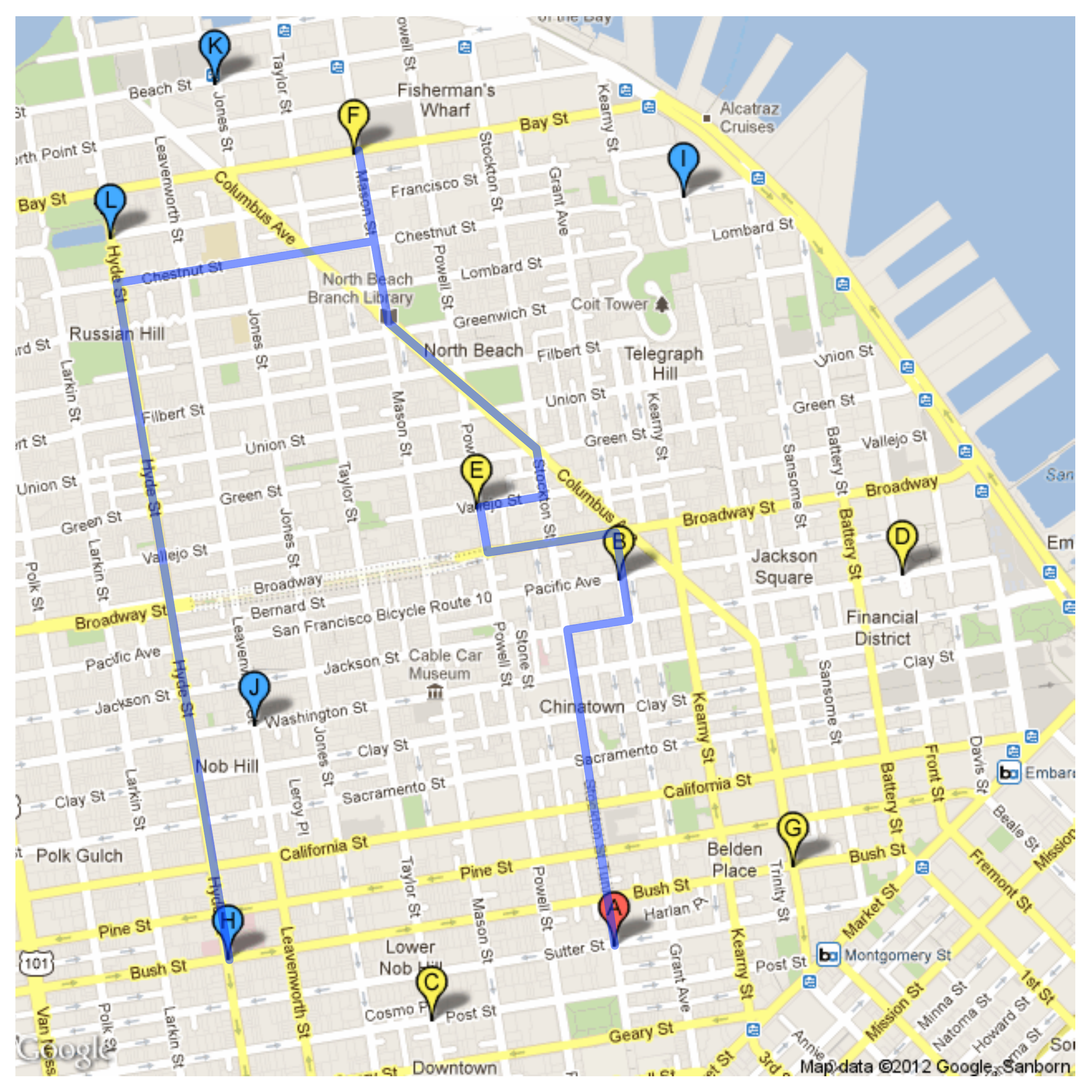}
\caption{The walk $S_4 = (\mathrm{A,\; B,\; E,\; F,\; H})$.}
\label{fig:TheS4}
\end{subfigure}
\caption{The four walks for patrolling San Francisco, corresponding to Table \ref{tab:locationDetails}. Location A (red) is $V_0$.  Locations B, C, D, E, F, G (yellow) correspond to $V_1$.  Locations H, I, J, K, L correspond to $V_2$.  The final monitoring walk is given by $[S_1,S_2,S_3,S_4,S_1,\ldots]$.}
\label{fig:crime_tours}
\end{figure}

The walk $\Delta(S)$ with $S=[S_1,S_2,S_3,S_4]$ is our final patrolling route. The latencies and costs induced by the intersections are shown in Table \ref{tab:S}. The expected number of crimes that occur between two consecutive visits to any intersection is bounded by $C(\Delta(S))=0.102$.


\begin{table} \footnotesize \begin{center} \begin{tabular}{|c|c|c|} \hline Index & Latency (seconds) & Cost (crimes per visit) \\ \thickhline
      A 	& 1159 	& 0.059\\
      B 	& 2193 	& 0.075 \\
      C 	& 2136 	    & 0.072 \\
      D 	& 2309 	& 0.076 \\
      E 	& 2694 	& 0.085 \\
      F 	& 2339 	& 0.074 \\
      G 	& 2779 	& 0.078 \\
      \textbf{H} 	& \textbf{4206} 	& \textbf{0.102} \\
      I 	& 4206 	& 0.077 \\
      J 	& 4206 	& 0.069 \\
      K 	& 4206 	& 0.061 \\
      L 	& 4206 	& 0.054 \\
      \hline \end{tabular} \caption{The latencies and costs induced by $S$ for the crime case-study.  The maximum cost is attained at intersection H.}  \label{tab:S} \end{center} \end{table}

\section{Conclusions and Future Work}
\label{sec:conclusions}

In this paper, we considered the problem of planning a path for a robot to monitor a known set of features of interest in an environment. We represented the environment as a graph with vertex weights and edge lengths and we addressed the problem of finding a walk that minimizes the maximum weighted latency of any vertex. We showed several results on the existence and non-existence of optimal and constant factor approximation solutions. We then provided two approximation algorithms; an $O(\log n)$-approximation algorithm and an $O(\log \rho_G)$-approximation algorithm, where $\rho_G$ is the ratio between the maximum and minimum vertex weights.  We also showed via simulations that our algorithms scale to very large problems consisting of thousands of vertices.

There are several directions for the future work .  We continue to seek a constant factor approximation algorithm, independent of $\rho_G$.  We also believe that by adding some heuristic optimizations to the walks produced by our algorithms, we could significantly improve their performance in practice.  Finally, we are currently looking at ways to extend our results to multiple robots.  One approach we are pursuing is to equitably partition the graph such that the single robot solution can be utilized for each partition.

\vskip1em
\noindent \textbf{Acknowledgements.}  This work was supported in part by the Natural Sciences and Engineering Research Council of Canada (NSERC).

\bibliographystyle{plain}
\bibliography{IJRR-bib}

\end{document}